\documentclass[aps,prd,twocolumn,nofootinbib,superscriptaddress,longbibliography,amsfonts,amssymb,amsmath,titlepage]{revtex4-2}

\usepackage[usenames,dvipsnames]{color}
\usepackage{microtype}
\usepackage{IEEEtrantools}
\usepackage{amsthm}
\usepackage{enumitem}
\usepackage{aas_macros}
\usepackage{bm}
\usepackage{dcolumn}
\usepackage{latexsym}
\usepackage{rotating}
\usepackage{longtable}
\usepackage{url}
\usepackage[linktocpage]{hyperref}
\usepackage[scr=rsfs]{mathalpha}

\def\pd{\ensuremath{\partial}}
\usepackage[T1]{fontenc}

\let\eth\relax
\font\ec=ecrm0800 at 12pt
\def\thorn{\hbox{\ec\char'336}}
\def\eth{\hbox{\ec\char'360}}

\newcommand{\Lie}{\mathcal{L}}
\newcommand{\GHPLie}{\text{\L}}
\def\GHPwt{\ensuremath{\circeq}}

\def\Sh{\ensuremath{{}_s S_{\ell m \omega}}}
\def\Rh{\ensuremath{ {}_s R_{\ell m \omega }}}
\newcommand{\Rq}[1]{\ensuremath{ {}_s R_{\ell m #1 }}}

\def\Kh{\ensuremath{ {}_s  K_{\ell m \omega}}}
\def\Ih{\ensuremath{ {}_s I_{\ell m \omega}}}

\def\Us{\ensuremath{\Upsilon_{s}}}
\def\tUs{\ensuremath{\hat \Upsilon_{s}}}

\newcommand{\R}{{\mathbb R}}
\newcommand{\CC}{{\mathbb C}}

\newcommand{\llangle}{\langle\langle}
\newcommand{\rrangle}{\rangle\rangle}

\def\sm{\ensuremath{q}} % spatial metric
\def\nv{\ensuremath{\nu}} % normal vector
\def\pU{\ensuremath{\varpi}}

\newtheorem{theorem}{Theorem}
\newtheorem{corollary}[theorem]{Corollary}
\newtheorem{lemma}[theorem]{Lemma}
\newtheorem{definition}{Definition}

\renewcommand{\Re}{\operatorname{Re}}
\renewcommand{\Im}{\operatorname{Im}}

\newcommand{\dd}{{\rm d}}
\newcommand{\cX}{\mathcal X}

\newcommand{\half}{\dfrac{1}{2}}
\newcommand{\abs}[1]{\left\lvert #1 \right\rvert}

\newcommand{\intd}[1]{\textrm{d} #1}

\newcommand{\GHPw}[2]{\left\{ #1, #2 \right\}}

\begin{document}

\title{Conserved currents for Kerr and orthogonality of quasinormal modes}

\author{Stephen R. Green}
\email{stephen.green@aei.mpg.de}
\affiliation{Max Planck Institute for Gravitational Physics (Albert Einstein Institute)\\
  Am M\"uhlenberg 1, D-14476 Potsdam, Germany}

\author{Stefan Hollands}
\email{stefan.hollands@uni-leipzig.de}
\affiliation{Institut f\"ur Theoretische Physik, Universit\"at Leipzig\\
  Br\"uderstrasse 16, D-04103 Leipzig, Germany}
\affiliation{Max Planck Institute for Mathematics in the Sciences, Inselstrasse 16\\
  D-04109 Leipzig, Germany}

\author{Laura Sberna}
\email{laura.sberna@aei.mpg.de}
\affiliation{Max Planck Institute for Gravitational Physics (Albert Einstein Institute)\\
  Am M\"uhlenberg 1, D-14476 Potsdam, Germany}
  
\author{Vahid Toomani}
\email{vahid.toomani@uni-leipzig.de}
\affiliation{Institut f\"ur Theoretische Physik, Universit\"at Leipzig\\
  Br\"uderstrasse 16, D-04103 Leipzig, Germany}

\author{Peter Zimmerman}
\email{zimmerator@protonmail.com}

\begin{abstract}
  
  We introduce a bilinear form for Weyl scalar perturbations of Kerr. The form is symmetric and conserved, and we show that, when combined with a suitable renormalization prescription involving complex $r$ integration contours, quasinormal modes are orthogonal in the bilinear form for different $(l,m,n)$. These properties are apparently not evident consequences of standard properties for the radial and angular solutions to the decoupled Teukolsky relations and rely on the Petrov type D character of Kerr and its $t$--$\phi$ reflection isometry.
  We show that quasinormal mode excitation coefficients are given precisely by the projection with respect to our bilinear form. These properties can make our bilinear form useful to set up a framework for nonlinear quasinormal mode coupling in Kerr. {We also provide a general discussion on conserved local currents and their associated local symmetry operators for metric and Weyl perturbations, identifying a collection containing an increasing number of derivatives.}

%   Along the way, \new{we include a general discussion on conserved local currents and their associated local symmetry operators for metric and Weyl perturbations of Kerr, which contain an increasing number of derivatives.}
  
  %we include a general discussion on conserved local currents and their associated local symmetry operators for metric and Weyl perturbations of Kerr. In particular, we obtain an infinite set of conserved, local, gauge invariant currents associated with Carter's constant for metric perturbations, containing $2n+9$ derivatives.
\end{abstract}

\maketitle
%\tableofcontents

%%%%%%%

\section{Introduction}
Quasinormal ringing is the principal gravitational-wave signature of
the final black hole after a binary merger. This is described by a
spectrum of complex quasinormal frequencies $\omega_{lmn}$, which are
uniquely specified in linear perturbation theory by the mass and spin
of the Kerr background\footnote{
{We are assuming here the applicability of the no-hair theorems; see, e.g., \cite{chrusciel6112living}.} 
}~\cite{KokkotasSchmidt1999,Nollert:1999ji,Berti2009}. Precise measurement of these frequencies
therefore characterizes the background~\cite{Echeverria89} and moreover constrains
deviations from general relativity (with more than one mode, or when
combined with other measurements)~\cite{Dreyer:2003bv,Berti:2005ys,Brito:2018rfr,LIGOScientific:2021sio}. Although data today already hint at
modes beyond the fundamental~\cite{Isi:2019aib,Cotesta:2022pci,Finch:2022ynt,Capano:2021etf}, future observations with sensitive
detectors are sure to enable detailed spectroscopy~\cite{Berti:2005ys,Bhagwat:2021kwv,Ota:2019bzl}.

To interpret future observations, however, it will be necessary to
understand quasinormal mode interactions. The ringdown follows a highly
nonlinear phase (the merger) and although numerical calculations indicate that a sum of
modes may be sufficient to represent the gravitational-wave emission~\cite{Giesler:2019uxc,Mourier:2020mwa,Chen:2022dxt},
it is not clear that this corresponds to a full nonlinear
description. Indeed, nonlinear ringdown effects have been identified in numerical simulations of binary mergers~\cite{Mitman:2022qdl,Cheung:2022rbm} as well as in anti-de Sitter black holes~\cite{Bantilan:2012vu,Sberna:2021eui}.
In other contexts (e.g., perturbations of large
anti-de Sitter black holes) quasinormal modes can interact and even
become turbulent~\cite{Green:2013zba,Adams:2013vsa}. The point of this paper is to introduce some tools that may be helpful when 
developing a theory of quasinormal mode interactions. 

Compared to normal modes, quasinormal modes do not in general form in a 
straightforward sense a
complete ``basis'' of solutions to the linearized field equations. In fact, black
hole perturbations are only described by quasinormal modes for an
intermediate time period in their evolution; at early times they are
described by a free propagation piece, and at
late times by a power law tail~\cite{Price1972a,Leaver1986b,Ching:1995tj}. 
The spatial wavefunction of a decaying quasinormal mode also
\emph{diverges} at the bifurcation surface and at spatial
infinity. This makes it hard to write down canonical (conserved) $L^2$-type inner products based on the usual Cauchy-surfaces
of Kerr.\footnote{Note however that one may choose hyperboloidal slices \cite{Zenginoglu:2011jz,PanossoMacedo:2019npm,Ripley:2022ypi,gajic2021quasinormal}; see the conclusions for a discussion of this alternative in connection with our approach.}.
Without an inner product, it is not 
clear how to project onto quasinormal modes to study nonlinear
mode mixing.

The main goal of this paper is to point out an unconventional bilinear form which may take
the place, for some purposes, of an inner product on quasinormal modes of Kerr. Before 
we introduce this notion, we develop a general theory for conserved -- under time evolution --
bilinear forms for Weyl scalars or metric perturbations. Similar to \cite{carter1977killing,Toth:2018ybm}, the key idea is 
to start with a ``Klein-Gordon'' type current for Weyl scalars or metric perturbations and 
to apply symmetry operators to the entries of this bilinear expression.  As we show, in Kerr spacetimes, 
such symmetry operators include, besides the obvious ones descending from the 
Killing symmetry, also an infinite tower of operators built from Carter's Killing tensor. 
(For the Weyl scalars, the symmetry operator of lowest differential order has two derivatives; 
for metric perturbations, it has six derivatives). 
In particular, using a combination of such operators we find an infinite set of new conserved, local, gauge invariant current associated with Carter's constant \cite{Carter:1968ks} in Kerr.\footnote{
\label{footnote1}
For an explanation of the relation with previous works \cite{carter1977killing, carter1979generalized, grant2020class, grant2020conserved, andersson2015spin,aksteiner2019symmetries}, see section \ref{sec:Symmetry}.}

The bilinear form of main interest for this paper is, however, not obtained from such differential symmetry operators but rather the symmetry operator associated with the discrete $t$--$\phi$ reflection.  We show that gravitational quasinormal modes
with different frequencies are orthogonal with respect to this
bilinear form. For the reader interested in the main result, the bilinear form is presented explicitly for quasinormal modes in \eqref{eq:mode-bilinear}. We show furthermore that the quasinormal mode excitation coefficients of a solution are given precisely by the projection of data onto the corresponding modes via the bilinear form.

The plan of this paper is as follows. In section \ref{sec:Bilinear} we recall the standard recipe for 
constructing conserved bilinear forms for partial differential operators. In section \ref{sec:Symmetry}
we introduce symmetry operators (including symmetry operators related to the Killing tensor, see also footnote \ref{footnote1}) 
to construct further conserved bilinear forms, and currents.
 In section \ref{sec:bilinear_tphi} we construct the bilinear form  $\llangle \cdot , \cdot \rrangle$ 
 using the $t$--$\phi$ reflection symmetry, which gives orthogonality of quasinormal modes in section \ref{sec:Ortho}.
Finally, in section \ref{sec:Lap transform} we explain the 
relation with excitation coefficients. Some technical aspects of this paper are deferred to various appendices.

\section{Bilinear form -- basic construction}
\label{sec:Bilinear}

Consider a partial differential operator $\cX$ acting on sections of some vector bundle, $E$, 
over a manifold $M$. We assume that $M$ is equipped with a volume form, 
$\epsilon_{a_1 \dots a_n}$; later we will always have a metric $g_{ab}$, so the volume form 
is chosen as the one compatible with the metric. Let $\tilde E$ be the dual vector bundle, i.e., 
each fibre is given by the $\CC$-linear maps of the corresponding fiber of $E$. If $\psi$ is a 
section of $E$ and $\tilde \psi$ is a section of $\tilde E$, we can pointwise form the scalar 
$\tilde \psi \psi \in \CC$. The formal adjoint is the unique differential operator $\cX^\dagger$
defined by the formula 
\begin{equation}
(\cX^\dagger \tilde \psi) \psi - \tilde \psi \cX \psi = \nabla_a x^a[\tilde \psi, \psi], 
\end{equation}
where $x^a[\tilde \psi, \psi]$ is local, i.e., at any point built from finitely many 
derivatives of the fields at that point. The divergence operator on the right is defined by our 
volume form and if it comes from a metric, as we assume from now, it 
is equal to the usual covariant derivative operator. Said differently, $\cX^\dagger \tilde \psi$ is 
obtained by the usual ``partial integration'' procedure dropping surface terms as if the above 
equation were placed under an integral sign. Note that, by contrast to quantum mechanics, 
$\dagger$ as defined above is $\CC$-linear, rather than anti-linear.

Now let $(\tilde \psi, \psi)$ be a pair of solutions to $\cX \psi = 0 = \cX^\dagger \tilde \psi$, and let 
$\Sigma$ be a codimension 1 submanifold of $M$ (later to be chosen as a constant $t$ slice of Kerr). 
Then, by Gauss' theorem, if $\tilde \psi$, $\psi$ have sufficient decay on $\Sigma$ for the following integral to be suitably convergent (e.g., if they are compactly supported), 
then the bilinear form 
\begin{equation}
\label{Xdef}
X[\tilde \psi, \psi] := \int_\Sigma x^a [\tilde \psi, \psi] \, \dd \Sigma_a \equiv \int_\Sigma (\star x) [\tilde \psi, \psi]
\end{equation}
is unchanged under local deformations of $\Sigma$, and we say that it is ``conserved''.
(Here $\star$ denotes the Hodge dual.) 
As a simple example, consider $\cX = \nabla^a \nabla_a - m^2$, the Klein-Gordon operator acting on real-valued functions $\psi$, 
so $E = \tilde E = \R$ is the trivial bundle. Then $\cX^\dagger = \cX$ and $x^a = - \tilde \psi \nabla^a \psi + \psi \nabla^a \tilde \psi$ is the Klein-Gordon (symplectic) current, 
which is of course conserved for any pair of solutions. 
The bilinear form in this case is just the symplectic form for Klein-Gordon theory. It is anti-symmetric 
under $\tilde \psi \leftrightarrow \psi$, but note that in the general case we cannot say that about the bilinear form since the bundles $E$
and $\tilde E$ cannot usually be identified in a natural way.   

As a second example, let $\mathcal{E}$ be the linearized Einstein operator on a Ricci-flat spacetime. It acts on symmetric covariant rank-2
tensors $h_{ab}$, so $E$ is equal to ${\rm Sym}(T^*M \otimes T^*M)$ in this case, and the dual bundle $\tilde E$ corresponds to symmetric contravariant 
rank-2 tensors, ${\rm Sym}(TM \otimes TM)$. The formula is
\begin{align}\label{eq:linearE}
  \mathcal{E}_{ab}(h) \equiv \frac{1}{2}\big[ &-\nabla^c\nabla_c h_{ab} - \nabla_a\nabla_bh + 2 \nabla^c\nabla_{(a} h_{b)c} \nonumber\\
 & + g_{ab}(\nabla^c \nabla_c h - \nabla^c\nabla^d h_{cd}) \big],
\end{align}
and under the identification of $E$ with $\tilde E$ (by using the metric $g^{ab}$ to raise indices), we have $\mathcal{E}^\dagger = \mathcal{E}$.
As in the Klein-Gordon case, this last relation follows because the linearized 
Einstein equation arises from an action principle. By explicit calculation, the boundary term $w^a \equiv x^a[\tilde h, h]$
is given by \cite{iyer1994some}
\begin{equation}
\label{wadef}
w^a = 
p^{abcdef}\left( h_{bc} \nabla_d \tilde h_{ef} - \tilde h_{bc} \nabla_d h_{ef}\right),
\end{equation}
where
\begin{align}
  p^{abcdef} = &g^{ae}g^{fb}g^{cd} - \frac{1}{2}g^{ad}g^{be}g^{fc} - \frac{1}{2}g^{ab}g^{cd}g^{df}\nonumber\\ 
  &- \frac{1}{2}g^{bc}g^{ae}g^{fd} + \frac{1}{2}g^{bc}g^{ad}g^{ef}.
\end{align}
The bilinear form 
\begin{equation}
\label{Wdef}
W[\tilde h, h] = \int_\Sigma
p^{abcdef}\left( h_{bc} \nabla_d \tilde h_{ef} - \tilde h_{bc} \nabla_d h_{ef}\right) \dd \Sigma_a ,
\end{equation}
is the symplectic form of General Relativity \cite{iyer1994some}.

Our third, and most important, example concerns the Teukolsky operator(s) for the perturbed Weyl scalars of the Kerr spacetime $(M,g_{ab})$, 
to which we will restrict attention from now on. For this, we shall employ the GHP
formalism~\cite{Geroch:1973am,Bini:2002jx,Aksteiner:2010rh,Toth:2018ybm} in the following, and we now briefly review the essential portions of this formalism which 
simplifies and also conceptualizes many calculations in the
Kerr -- or more generally, Petrov type D -- geometry. $l^a$ and $n^a$ are taken to be the repeated
principal null directions which are completed to a null tetrad by defining a smooth pair of
complex null rays $(m^a, \bar m^a)$ that span the remaining
dimensions. We choose the normalization $l_an^a=1$ and
$m_a\bar m^a=-1$, corresponding to the $-2$ signature. The metric
then takes the form
\begin{equation}\label{eq:NP met}
g_{ab} = 2l_{(a}n_{b)}-2m_{(a}\bar m_{b)}.
\end{equation}
The basic idea is to contract any tensor field on $M$ into
the legs of the Newman-Penrose (NP) tetrad $(l^a, n^a, m^a, \bar m^a)$
in all possible ways\footnote{We do not require tensor fields to be
  \emph{fully} contracted with the tetrad, so in general we refer to
  NP \emph{tensors}, not just scalars. In other words, there can
  remain tensor indices after contraction.} and to represent the
action of the covariant derivative operator $\nabla_a$ in terms of
these tetrad components, in a way that preserves a natural grading by
spin and boost weights.

Fields $\eta$ obtained by contracting with the tetrad are
classified according to their spin and boost weights as follows. Under
a local rotation that preserves the real null pair, the tetrad
transforms as $(l^a, n^a, e^{i\Gamma} m^a, e^{-i\Gamma} \bar m^a)$,
whereas under a local boost that preserves the directions of the real
null pair, it transforms as
$(\Lambda l^a, \Lambda^{-1}n^a, m^a, \bar m^a)$, where $\Lambda$,
$\Gamma$ are smooth real-valued functions.  If we combine these
functions into the complex function $\lambda^2 = \Lambda e^{i\Gamma}$,
then $\eta$ is said to possess (real) GHP weights $(p,q)$ if under
the above combined local rotation and boost of the tetrad, it
transforms as
\begin{equation}\label{trafo}
\eta \to \lambda^p \bar \lambda^q \eta. 
\end{equation}
We write $\eta \GHPwt (p,q)$ if this is the case. In the GHP
formalism, only quantities with the same weight may be added, whereas
weights behave additively under multiplication.

From the mathematical viewpoint, the GHP formalism can be understood in terms of principal fibre bundles and
their associated vector bundles, as follows. Consider the set of oriented null frames aligned with the given null directions. On each such frame, we may pointwise 
perform a boost/rotation, which as we described can be combined into a nonzero complex number $\lambda \in \CC_\times$. Thus, we have a multiplicative action of 
$\CC_\times$ on the set of frames which gives this set the structure of a principal $G$-bundle: A principal bundle is abstractly a bundle $P$ over $M$ such that a group $G$
can act by right multiplication $X \to X \cdot g$ in the fibre -- in our case $X$ is an NP frame aligned with the principal null direstions and $g \leftrightarrow \lambda$. Given a principal $G$-bundle and a representation $R$ of $G$ on some vector space $V$, there is a canonical construction of an ``associated'' vector bundle. The sections of this bundle correspond physically to quantities defined on $M$ that ``transform in the representation $R$''.
More precisely, the elements in this associated bundle are the equivalence classes of pairs $(X,v)$ where $X \in P$ 
and $v \in V$ where $(X,v)$ is declared to be equivalent to 
$(X \cdot g, R(g)v)$. In the present example, $R_{p,q}(\lambda)v = \lambda^p \bar \lambda^q v$ and $V = \CC$, which corresponds precisely to the ``transformation law'' \eqref{trafo}. The associated vector bundle is denoted in general by $P \ltimes_R V$ and its fibres are isomorphic to $V$. In our case, we get 1-dimensional complex (``line'') bundles $L_{p,q}=P \ltimes_{p,q} \CC$ over $M$ labelled by the GHP weights $(p,q)$. The number $s=\frac{1}{2}(p-q)$ is commonly referred to as the spin. Of course, we could tensor $L_{p,q}$ with the usual tensor bundles $T^{(r,s)} M$ to host objects that have GHP weights and tensor indices at the same time such as $l^a$ or $R_{abcd} m^a m^d$.

The advantage of the above invariant viewpoint involving associated vector bundles is that we can naturally see what quantities are defined in a frame independent manner, which quantities can naturally be added, etc. This provides not only an extremely useful guiding principle in the -- usually very complicated -- calculations related to Kerr, but also means that one is always intrinsically dealing with objects that behave in a well-defined manner under a change of frame. To make the formalism really useful, one needs covariant derivative operators on the bundles $L_{p,q}$. These are given by
\begin{equation}
\Theta_a = \nabla_a - \tfrac{1}{2} (p-q) \bar m^b \nabla_a m_b -\tfrac{1}{2} (p+q) n^b \nabla_a l_b.
\end{equation}
The Teukolsky operators also feature the ``gravito-magnetic potential'' which is  given by
\begin{equation}
\label{Bdef}
B^a \equiv -(\rho n^a - \tau \bar m^a) \GHPwt (0,0), 
\end{equation}
where $\rho, \tau$ are related to spin-coefficients~\cite{Geroch:1973am,Bini:2002jx,Aksteiner:2010rh}; see appendix \ref{app:D}.
The Teukolsky operator acts on GHP-scalars of the same weight\footnote{For the definition of $\mathcal O$ and 
  $\mathcal{O}^\dagger$ for general GHP weights see appendix \ref{app:B}.} as the perturbed Weyl scalar $\psi_0$, 
i.e., $(p,q) = (4,0)$ and is given by 
\begin{equation}
    \mathcal{O} = g^{ab}(\Theta_a + 4 B_a)(\Theta_b + 4 B_b) - 16 \Psi_2 
\end{equation}
  with $\Psi_2$ a background Weyl-scalar. So $E=L_{4,0}$ now. Since the dual vector bundle to $L_{p,q}$ is $L_{-p,-q}$, the adjoint 
  Teukolsky operator $\mathcal{O}^\dagger$ acts on GHP scalars of weight $(-4,0)$. It is given by
 \begin{equation}\label{eq:Odagger}
\mathcal{O}^\dagger  = g^{ab}(\Theta_a - 4 B_a)(\Theta_b - 4 B_b) - 16 \Psi_2  .
\end{equation} 
 It follows that the boundary term $x^a[\tilde \Upsilon, \Upsilon] \equiv \pi^a$ (with $\tilde \Upsilon \GHPwt (4,0), \Upsilon \GHPwt (-4,0)$) 
 is given in the case of the Teukolsky operator by 
 \begin{equation}
 \label{pidef}
   \pi^a =  \tilde \Upsilon(\Theta^a - 4B^a)\Upsilon - \Upsilon (\Theta^a + 4 B^a) \tilde \Upsilon 
 \end{equation}
We denote the corresponding bilinear form -- formally similar to the Klein-Gordon inner product of a charged scalar field -- by 
\begin{equation}
\label{Pidef}
\Pi[\tilde \Upsilon, \Upsilon] = \int_\Sigma \left[ \tilde \Upsilon(\Theta^a - 4B^a)\Upsilon - \Upsilon (\Theta^a + 4 B^a) \tilde \Upsilon \right] \dd \Sigma_a.
\end{equation}

The Teukolsky equation/operator and the linearized Einstein equation/operator are well-known to be related and this implies that the 
bilinear forms $W$ and $\Pi$ as in \eqref{Wdef} and \eqref{Pidef} are related, too. \cite{Prabhu:2018jvy} have shown that for $\Upsilon$ a
smooth solution to $\mathcal O^\dagger \Upsilon=0$ arising from compact
support data and $h_{ab}$ a smooth solution to
$\mathcal E h_{ab} = 0$, an identity of the following form holds
\begin{equation}\label{eq:intertwine_inf}
w^a[h, \mathcal S^\dagger \Upsilon] = - \pi^a[\mathcal T h, \Upsilon] + \nabla_b H^{ab}[\Upsilon, h], 
\end{equation}
where $H^{ab}$ is a skew symmetric local tensor. Furthermore~\cite{Aksteiner:2014thesis,Araneda:2016iwr}
\begin{subequations}
  \begin{align}
    \mathcal{S}(T) &= Z^{bcda} (\Theta_a + 4 B_a) \Theta_b T_{cd},\\
    \mathcal{T}(h) &= - \frac{1}{2}Z^{bcda} \Theta_a \Theta_b h_{cd},
  \end{align}
\end{subequations}
where $Z^{abcd} \equiv Z^{ab}Z^{cd}$, and $Z^{ab} \equiv 2l^{[a}m^{b]}$, are operators such that the Teukolsky-Wald identity holds:
\begin{equation}\label{SEOT}
\mathcal S \mathcal E = \mathcal O \mathcal T.
\end{equation}
This equation encodes that the action $\mathcal T (h)$ on a metric perturbation $h_{ab}$ -- which equals the perturbed Weyl scalar $\psi_0$ --
gives a solution to Teukolsky's equation $\mathcal{O} \psi_0 = 0$. Conversely, taking an adjoint of \eqref{SEOT}, i.e., 
$\mathcal{E} \mathcal{S}^\dagger = \mathcal{T}^\dagger \mathcal{O}^\dagger$, shows that any solution $\mathcal{O}^\dagger \Upsilon = 0$ of 
GHP weight $(-4,0)$ (``Hertz potential'') is such that $h_{ab} = \Re \mathcal{S}^\dagger_{ab} \Upsilon$ is a solution to the linearized Einstein equations. 

Ref.~\cite{Prabhu:2018jvy} did not derive the explicit form for $H^{ab}$ but 
argued for the above equation \eqref{eq:intertwine_inf} to hold on general grounds based on \eqref{SEOT}. The main use of the above identity 
\eqref{eq:intertwine_inf} is to 
relate the corresponding bilinear forms $W[h, \mathcal S^\dagger \Upsilon]$ and $\Pi[\mathcal T h, \Upsilon]$ for a Cauchy surface $\Sigma$
of the exterior of Kerr. This identity is obtained by simply integrating the above identity over $\Sigma$. If all 
fields are falling off rapidly at the horizon and spatial infinity, then the boundary term arising from $H^{ab}$ will not contribute; 
in other cases, $H^{ab}$ will contribute surface terms. Their computation is fairly long and non-trivial and therefore deferred to appendix \ref{app:A}. 
If $\Sigma$ is a co-dimension one surface with boundary $\partial \Sigma$,
$\Upsilon$ is a
smooth solution to $\mathcal O^\dagger \Upsilon=0$ and $h_{ab}$ a smooth solution to
$\mathcal E h_{ab} = 0$, then we have
\begin{equation}\label{eq:intertwine}
  W[h, \mathcal S^\dagger \Upsilon] = - \Pi[\mathcal T h, \Upsilon] + B[h,\Upsilon] 
\end{equation}
where $B = \int_{\partial \Sigma} H^{ab} \dd \Sigma_{ab}$. When $\Sigma$ is 
a slice of constant $t$ in Boyer-Lindquist coordinates, $\partial \Sigma$ would correspond to the bifurcation surface at $r=r_+$ and the sphere at $r=\infty$.
Using this formula, the reader can readily transfer results on 
bilinear forms in this paper between the metric perturbation and Teukolsky variables. 

\section{Bilinear forms from infinitesimal symmetry operators}
\label{sec:Symmetry}

Consider again a general  partial differential operator $\cX$ acting on sections of some vector bundle, $E$, 
over a manifold $M$. We have the corresponding conserved bilinear form $X[\tilde \psi, \psi]$ defined 
by \eqref{Xdef}. Now suppose $\mathcal{C}$ is a partial differential operator acting on $E$ mapping 
solutions to $\cX \psi = 0$ to solutions -- this is equivalent to the statement that there is a partial differential operator $\mathcal{D}$
such that $\cX \mathcal{C} = {\mathcal D} \cX$. Such an operator is called a ``symmetry operator''. The symmetry operators form an algebra which is trivial for a generic operator $\cX$.
If we have a symmetry operator, then 
$X[\tilde \psi, \mathcal{C} \psi]$ is also a conserved 
bilinear form, i.e., invariant under local changes of the surface $\Sigma$ in \eqref{Xdef}, see e.g. \cite{carter1977killing, carter1979generalized} for a similar observation.

Let us apply this recipe to the linearized Einstein operator $\mathcal{E}$ on the Kerr spacetime. The Kerr spacetime
has two Killing vector fields, $t^a, \phi^a$ corresponding to asymptotic time translations and rotations. The Lie
derivatives $\mathcal{L}_t, \mathcal{L}_\phi$ evidently commute with $\mathcal{E}$ and thus provide two conserved quadratic forms:
\begin{equation}
\label{canen}
E[h] = W[h,\mathcal{L}_t h], \quad J[h] = W[h,\mathcal{L}_\phi h].
\end{equation}
They correspond to the canonical energy and canonical angular momentum of the perturbation $h_{ab}$ when $\Sigma$ is a Cauchy surface
stretching between the bifurcation surface and spatial infinity \cite{hollands2013stability}. 

If we want to repeat a similar construction for the Teukolsky operator $\mathcal{O}$ and the corresponding bilinear form $\Pi$ 
we face the problem that the Lie-derivative 
in general is not well-defined on an arbitrary vector bundle (though it is on the usual bundles of tensors over $M$). 
In the GHP formalism, the vector bundles $L_{q,p}$ in question 
are defined relative to an NP tetrad, and in such a case we can still give a definition of the Lie derivative along a Killing vector
field, though not an arbitrary vector field, as we now describe. The point is that if $g_{ab}$ has an isometry $\varphi$ that
  preserves the globally defined null directions, then this 
  constitutes an intrinsically
  defined action on GHP tensors $\eta \GHPwt (p,q)$.  More
  explicitly, if $\varphi$ preserves the null directions, then it must
  be the case that it acts on a given null frame as
  $\varphi_* l^a = \Lambda l^a$, $\varphi_* n^a = \Lambda^{-1} n^a$,
  and $\varphi_* m^a = e^{i\Gamma} m^a$, for some real functions $\Lambda$,
  $\Gamma$ on $M$ that depend on the chosen frame and
  $\varphi$. The action of $\varphi$ on $\eta$ is then invariantly
  defined since GHP tensors are functionals of the
  null tetrads giving rise to the prescribed pair of null
  directions. In the given null frame, this action amounts to
  $\varphi^{\text{GHP}}_*\eta \equiv \lambda^{-p} \bar \lambda^{-q}
  \varphi_*\eta$, where $\lambda^2 = \Lambda e^{i\Gamma}$ and
  $\varphi_*$ is the standard pushforward on functions 
  (or tensors). In particular, the
  tetrad vectors are invariant under $\varphi_*^{\text{GHP}}$.

  Infinitesimally, if $\varphi_t$ is a 1-parameter group of transformations generated by
  a Killing field $\chi^a$ with corresponding
  $\lambda_t$, then the corresponding ``Lie'' transport of $\eta \GHPwt (p,q)$ is given
  by~\cite{edgar2000integration}
  \begin{IEEEeqnarray}{rClCl}\label{eq:GHPLie}
    \GHPLie_\chi \eta &=& \lim_{t\to0}\frac{(\varphi_{-t})^{\text{GHP}}_\ast \eta - \eta}{t} && \nonumber\\
                      &=& (\mathcal{L}_\chi - p w - q \bar w) \eta &\GHPwt& (p,q),
  \end{IEEEeqnarray}
  in the given frame. Here, $\mathcal{L}$ denotes the standard Lie derivative, and 
  \begin{align}
    w &= \frac{d}{dt} \log \lambda_t \bigg|_{t=0}\\
      &= \frac{1}{2} \left(n_a\mathcal{L}_\chi l^a - \bar m_a \mathcal{L}_\chi m^a \right).
  \end{align}
  If we introduce the bivector 
  $Y \equiv n \wedge l - \bar m \wedge m$ (for further details on the
  bivector calculus see, e.g.
  \cite{fayos1990electromagnetic,Aksteiner:2014thesis}) and use
  the fact that $\chi^a$ is a Killing field, so
  $\nabla_{(a}\chi_{b)} = 0$, then~\eqref{eq:GHPLie} can be manipulated to obtain
  \begin{equation}\label{eq:GHPLie-simplified}
    \GHPLie_\chi \eta = \left[ \mathcal{L}^\Theta_\chi
      - \frac{p}{4} Y^{ab}\Theta_a\chi_b
      - \frac{q}{4} \left( Y^{ab} \Theta_a\chi_b \right)^\ast \right] \eta,
  \end{equation}
  where $\mathcal{L}^\Theta$ is the standard Lie derivative with $\nabla_a$
  derivatives replaced by $\Theta_a$ derivatives. In this notation, the GHP Lie derivative 
  is also defined for GHP-tensors, i.e., 
  sections in a bundle $L_{p,q}$ tensored with 
  $TM$ or $T^*M$. In any case, the GHP Lie
  derivative defined here is manifestly GHP covariant, and it can be
  checked that it satisfies the Leibniz rule. The expression for  
  $\GHPLie_\chi$ 
  in a chosen NP tetrad will depend on that choice. 
  For the Kinnersley tetrad \eqref{eq:Kintet}, $w=0$, but
  $w$ can be different from zero for other choices of the frame.

  With these definitions, it then follows that $\GHPLie_\chi$ for $\chi^a$ either $t^a$ or $\phi^a$ 
  commutes with the covariant derivative $\Theta_a$ and annihlates $g_{ab}, n^a, l^a, m^a, \bar m^a, B_a$. Therefore, 
  $\GHPLie_\chi$ also commutes with the Teukolsky operators, 
  \begin{equation}
  [\GHPLie_\chi, \mathcal{O}] = 0 = [\GHPLie_\chi, \mathcal{O}^\dagger], \quad \chi^a = t^a , \phi^a, 
  \end{equation}
   and it thus defines a symmetry operator. {The corresponding conserved currents arising from $\pi^a$ \eqref{pidef} using the general construction described above have been discussed by \cite{Toth:2018ybm}.}
  
  There exist other symmetry operators in the Kerr (and more generally, Petrov type D-) spacetimes
  related to the Killing tensor $K_{ab}$ that exists in those spacetimes. The construction of those operators for spin $s=0, \tfrac{1}{2}$ in the Teukolsky equation goes back to 
  \cite{carter1977killing, carter1979generalized}; here we present the corresponding symmetry operator for arbitrary GHP-weights $(p,q)$. Similar operators have appeared also in 
  \cite{grant2020class, grant2020conserved}, eq. III.3, for spin $s=1, 2$, though not in the GHP covariant form presented here which makes manifest the relationship with the Killing tensor. This tensor 
  is given by 
  \begin{equation}
  \label{Kabdef}
  K^{ab} = - \dfrac{1}{4} \left( \zeta - \bar{\zeta} \right)^2 l^{(a} n^{b)} + \dfrac{1}{4} \left( \zeta + \bar{\zeta} \right)^2 m^{(a} \bar{m}^{b)}
  \end{equation} 
  where we use the shorthand
\begin{equation}
\label{zetadef}
\zeta = - \Psi^{-\tfrac{1}{6}}_2 \bar{\Psi}^{-\tfrac{1}{6}}_2 \rho^{-\tfrac{1}{2}} \bar{\rho}^{\tfrac{1}{2}} \circeq \GHPw{0}{0}
\end{equation}
with 
$\rho$ one of the spin coefficients in the GHP formalism. The desired symmetry operator $\mathcal{K}$
acting on GHP scalars of weights $(p,q)$ is defined as
% \begin{widetext}
\begin{align}\label{eq:Koperatordef}
\mathcal{K} \eta =&\left( \Theta_a + p B'_a + q \bar{B}'_a \right) K^{ab} \left( \Theta_b + p B'_b + q \bar{B}'_b \right) \eta  \nonumber\\
&+ 2 (p \gamma + q \bar{\gamma}) \GHPLie_\xi \eta
\end{align}
% \end{widetext}
where 
\begin{equation}
\label{xidef}
\xi_a = \zeta \left( B_a - B'_a \right), 
\end{equation}
is proportional to a Killing vector field, and $\gamma = ( \zeta^2 - \bar{\zeta}^2 )/(8 \zeta)$. Here and in the following, a prime as in $B'_a$ means the GHP priming operation $n^a \leftrightarrow l^a, m^a \leftrightarrow \bar m^a$.
In Boyer-Lindquist coordinates and the Kinnersly frame (see appendix \ref{app:D}), $\xi^a = M^{-1/3} t^a$, $\gamma =  M^{- 1/3} \frac{- i a \cos \theta}{2(r - i a \cos\theta)}$ and $\GHPLie_\xi \eta = M^{-1/3} \partial_t \eta$.
$\mathcal{K}$ is called a symmetry operator because one can show that 
\begin{equation}
\label{commutator}
[\mathcal{K}, \mathcal{O}] = 0 = [\mathcal{K}, \mathcal{O}^\dagger] 
\end{equation}
when acting on GHP quantities of weight $(4,0)$ or $(-4,0)$, respectively. The proof of this statement is rather 
nontrivial and deferred to appendix \ref{app:B}, where we also prove the commutation property for arbitrary $(p,q)$. It follows from the properties of the GHP Lie derivative that $[\GHPLie_\chi, \mathcal{K}]=0$ for any Killing vector field $\chi^a$, so we have:

\begin{theorem}
$\GHPLie_t, \GHPLie_\phi, \mathcal{K}$ generate a commutative, infinite-dimensional algebra of symmetry operators for Teukolsky's operator $\mathcal O$ for any GHP weights $(p,q)$.
\end{theorem}

Hence, by the general scheme, if we have  
solutions to $\mathcal{O} \tilde \Upsilon = 0 = \mathcal{O}^\dagger \Upsilon$,
and symmetry operators $\mathcal{A}, \mathcal{B}$, then the bilinear 
form $\Pi[\mathcal{A}\tilde \Upsilon, \mathcal{B} \Upsilon]$, with $\Pi$ as in \eqref{Pidef}, is conserved, i.e. unchanged under local deformations of the Cauchy surface $\Sigma$. We caution the reader that  such bilinear forms can be trivial, i.e., be equivalent to forms that are conserved identically; see appendix \ref{sec:trivial} for some discussion. 
 
It is possible to derive symmetry operators also for the linearized Einstein tensor $\mathcal{E}$ (and for the Maxwell equations) on Kerr or more generally, a Petrov type D spacetime. Let $n=0,1,2,\dots$ and set
\begin{equation}
\label{Cndef}
\mathcal{C}_n = \mathcal{S}^\dagger \mathcal{K}^n \zeta^{2s} \mathcal{T}',
\end{equation}
as well as
\begin{equation}
\mathcal{D}_n = \mathcal{T}^\dagger \mathcal{K}^n \zeta^{2s} \mathcal{S}',
\end{equation}
where for spin-2 considered here we should take $s=2$, and where we use the GHP priming operation. 
Then
\begin{equation}
\begin{split}
\mathcal{E} \mathcal{C}_n =& \mathcal{E} \mathcal{S}^\dagger \mathcal{K}^n \zeta^{2s} \mathcal{T}' \\
=& \mathcal{T}^\dagger \mathcal{O}^\dagger \mathcal{K}^n \zeta^{2s} \mathcal{T}'\\ 
=& \mathcal{T}^\dagger \mathcal{K}^n \mathcal{O}^\dagger \zeta^{2s} \mathcal{T}' \\
=& \mathcal{T}^\dagger \mathcal{K}^n \zeta^{2s} \mathcal{O}' \mathcal{T}'\\
=& \mathcal{T}^\dagger \mathcal{K}^n \zeta^{2s} \mathcal{S}' \mathcal{E}\\
=& \mathcal{D}_n \mathcal{E}
\end{split}
\end{equation}
where we used twice the Teukolsky-Wald identity \eqref{SEOT}, the commutation $[\mathcal{O}^\dagger, \mathcal{K}]=0$, as well as the intertwining 
relation $\mathcal{O}^\dagger \zeta^{2s} = 
\zeta^{2s} \mathcal{O}'$.
When acting on a perturbation $h_{ab}$, $\mathcal{T}'(h)$ gives the perturbed Weyl scalar $\psi_4$, which is gauge invariant. Therefore, we see that $\mathcal{C}_n(h)=0$ for any gauge perturbation $h_{ab} = \mathcal{L}_\xi g_{ab}$. 

By the results of appendix \ref{app:B} another symmetry operator for $\mathcal E$ would be $\mathcal{C}_n = \mathcal{S}^\dagger \mathcal{G}^n \zeta^{2s} \mathcal{T}'$, with $\mathcal{D}_n = \mathcal{T}^\dagger \mathcal{G}^{\dagger n} \zeta^{2s} \mathcal{S}'$ (with similar proof, see appendix \ref{app:B} for the definition of $\mathcal G$), and further symmetry operators are obtained by the GHP prime- and overbar operations applied to these $\mathcal{C}_n$'s. Finally, by putting $s=1$ in the above expressions, and defining $\mathcal{T}, \mathcal{S}$ so that the analog of the Teukolsky-Wald identity \eqref{SEOT} holds for electromagnetic perturbations, where $({\mathcal E}A)_a = \nabla^b \nabla_{[a} A_{b]}$, we get similar operators in the electromagnetic case.

As a consequence, in all cases, $\mathcal{C}_n$ give symmetry operators for $\mathcal{E}$ of order $4+2n$ for spin-2 and of order $2+2n$
for spin-1. Regarding our operator $\mathcal{C}_0$ for spin-2, we remark that a very similar looking operator has been considered by \cite{grant2020class}, Eq.~III.14. 
Regarding our operator $\mathcal{C}_1$, 
a similar looking operator has been considered in 
\cite{grant2020class}, Eq.~III.47 and also in 
\cite{aksteiner2019symmetries}, Thm.~16. However, closer inspection of the operator\footnote{\cite{grant2020class}, Eq. III.14 on the other hand is manifestly local.} in \cite{grant2020class}, Eq. III.47 shows that it is non-local, while our operators are all local and also manifestly GHP covariant. The relation of our operators $\mathcal{C}_1$ to the order 6 symmetry operator asserted in \cite{aksteiner2019symmetries} is not completely clear to us and the same goes for our other operators $\mathcal{C}_1'$, etc. For spin-1, symmetry operators of orders 2 and 4 have been discussed in \cite{grant2020conserved,andersson2015spin}, and the comparison to ours is qualitatively similar.\footnote{In the spin-0 case where ${\mathcal E} \phi = \nabla^a \nabla_a \phi$, the corresponding symmetry operators are powers ${\mathcal C}_n={\mathcal K}^n$ where 
$\mathcal{K} = \nabla_a K^{ab} \nabla_b$, which are of order $2n$ and have already been described.}

By the general theory, for example ($n=0,1,2, \dots$)
 \begin{equation}
\chi_{(n)}[h] = W[\overline{\mathcal{C}_0 h},  \mathcal{C}_n h]
\end{equation}
with $W$ as in \eqref{Wdef} are conserved for all solutions $h_{ab}$ to the linearized Einstein equations, i.e.~unchanged under local deformations of the Cauchy surface $\Sigma$. The corresponding conserved currents are 
 \begin{equation}
 \label{jdef}
j^a_{(n)} = w^a[\overline{\mathcal{C}_0 h},  \mathcal{C}_n h]
\end{equation}
with $w^a$ as in \eqref{wadef}. Note that each $j^a_{(n)}$ is 
local and gauge invariant from the properties of $\mathcal{C}_n$. The concrete expressions of $j^a_{(n)}$ 
are very long and contain $2n+9$ derivatives of $h_{ab}$.
For the reason explained below, we call $j^a_{(n)}$ the ``Carter current(s)''.

To gain some insight into the meaning of the conserved quantities $\chi_{(n)}$, we make a WKB (high frequency) analysis similar to \cite{green2016superradiant}, see also \cite{grant2020class}. If the momentum of the 
sharply collimated WKB wave packet $h_{ab}$ is $p_a$
and its amplitudes defined with respect to a suitable basis of polarization tensors are $A_{+,\times}$, the result is
\begin{equation}
\begin{split}
\label{chinint}
\chi_{(n)}[h] = & \int_\Sigma j^a_{(n)} \dd \Sigma_a \\
\sim &  \, \, -i (-1)^n \int_\Sigma p^a 
{\rm Im}(A_+ \bar A_\times) \times \\
& \qquad \times Q(p)^{n+4} \, \dd \Sigma_a
\end{split}
\end{equation}
where $K^{ab} p_a p_b = Q(p)$ denotes the  Carter constant. See appendix \ref{app:WKB} for more 
detail on the derivation of this formula 
and on the 
precise definitions of the WKB wave functions, polarizations, etc.

We can obviously form alternative conserved quantities by other combinations of the various symmetry operators of the linearized Einstein operator described above giving e.g., the GHP primed version of our Carter currents $j^{a \prime}_{(n)}$. We note also that such currents could have alternatively been constructed from $\pi^a$ \eqref{pidef}, taking $\Upsilon = \zeta^4 \psi_4$ and $\tilde \Upsilon = \psi_0$ and acting on those with various symmetry operators for the Weyl scalars, as described above. 

We finally remark that conserved currents  for metric perturbations related to Carter's constant have also been considered in \cite{grant2020class}, Eqs. IV.14-16. Eq. IV.14 is very similar to our $j_{(0)}^a$ but their currents Eqs. IV.15-16 are different from our Carter currents $j^a_{(n)}$ or their GHP primes because unlike ours, they are based on non-local currents requiring a mode decomposition of the solutions.

\section{Bilinear form from $t$--$\phi$ reflection}\label{sec:bilinear_tphi}

In the previous section, we combined the basic conserved bilinear form \eqref{Pidef} with symmetry operators, which arise in particular 
from the Killing vector fields of Kerr. One naturally expects that a similar construction should be possible for the discrete isometry of
Kerr, namely the $t$--$\phi$ reflection map $J: (t,\phi) \to (-t,-\phi)$ where here and in the following we refer to Boyer-Lindquist coordinates.
However, just as for Killing vectors, some care has to be taken when defining the action of $J$ on GHP scalars with nontrivial weights 
$(p,q)$. So we first turn to this issue.

The map $J$ swaps the null directions $l^a$ and $n^a$ and changes the
  orientation on the orthogonal complement of these null directions
  spanned by $m^a$, $\bar m^a$.  
  There must thus be $\Lambda$, $\Gamma$ depending on the null tetrad
  such that $J_* l^a = -\Lambda n^a$, $J_* n^a = -\Lambda^{-1} l^a$, and
  $J_* m^a = e^{i\Gamma} \bar m^a$, where we have defined $J$ to act on
  tensors by the push-forward. 
  By analogy with the previous
  case of isometries which are continuously deformable to the identity, 
  it is then natural to define for $\eta \GHPwt (p,q)$ a GHP
  reflection
  \begin{equation}
    \label{Jdef}
    \mathcal J \eta \equiv i^{p+q} \lambda^{-p} \bar \lambda^{-q} \eta \circ J \GHPwt (-p,-q) 
  \end{equation}
  in the given frame. 
  
  The operator $\mathcal J$ is evidently a GHP priming operation combined with $t \to -t, \phi \to -\phi$, and
 is therefore easily seen to be GHP covariant (i.e.~defined intrinsically as a map from sections in $L_{p,q}$ to sections in $L_{-p,-q}$, irrespective of the chosen frame), 
 but, by contrast to the ``pull-back'' arising from isometries continuously connected to the 
  identity as considered above, it changes the GHP weights. In this sense it is similar to the CPT operator arising in quantum field theory. 
  It is clear that $\mathcal J^2 = 1$ and one can relatively easily show the ``anti-commutation'' relations 
  $\GHPLie_{t} \mathcal J = - \mathcal J \GHPLie_{t}$, $\GHPLie_{\varphi} \mathcal J = - \mathcal J \GHPLie_{\varphi}$
  with the GHP Lie-derivative defined above.
 We also note an important intertwining property
of the $t$--$\phi$ reflection operator $\mathcal J$ with the Teukolsky operator and its adjoint, namely,
\begin{align}\label{eq:OJ}
  \mathcal O \Psi_2^{4/3} \mathcal J = \Psi_2^{4/3} \mathcal J \mathcal O^\dagger,
\end{align}
where we used basic properties of gravito-magnetic field $B_a$ and its GHP prime $B'_a$, as well as the relation
\begin{equation}\label{eq:gradPsi}
  \Theta_a \Psi_2 = -3 (B_a + B'_a) \Psi_2.
\end{equation}
 In the Kinnersley frame and Boyer-Lindquist coordinates (see appendix \ref{app:D}), the $\mathcal J$ operator corresponds to sending $t \to -t, \phi \to -\phi$
 and multiplication according to \eqref{Jdef} by appropriate powers of $\lambda, \bar \lambda$, where $\lambda$ is given in this case 
 explicitly by 
\begin{equation}\label{eq:boostParams}
  \lambda 
  = \sqrt{2} (r-ia \cos \theta) \Delta(r)^{-1/2}.
\end{equation}

We are now in a position to define the bilinear form. For simplicity, we restrict at first to entries having compact support on the 
Cauchy surface $\Sigma$ in order to avoid any convergence problems.

\begin{definition}[Bilinear form for compact support]
  Let $\Upsilon_1, \Upsilon_2 \GHPwt (-4,0)$ be smooth GHP scalars
  of compact support on $\Sigma$ in the kernel of
  $\mathcal O^\dagger$. Then we set
  \begin{equation}\label{bilinear}
    \llangle \Upsilon_1, \Upsilon_2 \rrangle \equiv \Pi_\Sigma[\Psi_2^{4/3} \mathcal J \Upsilon_1, \Upsilon_2]
  \end{equation}
  with $\Pi$ as in \eqref{Pidef}.
\end{definition}

\begin{lemma}\label{lemma:compactsupport}
  Under the conditions of the definition, we have
  \begin{enumerate}[label=(\roman*), start=1]
  \item $\llangle \Upsilon_1, \Upsilon_2 \rrangle$ is $\CC$-linear in both entries.
  \item
    $\llangle \Upsilon_1, \Upsilon_2 \rrangle=\llangle \Upsilon_2,
    \Upsilon_1 \rrangle$,
  \item
    $\llangle \GHPLie_t\Upsilon_1, \Upsilon_2 \rrangle=\llangle
    \Upsilon_1, \GHPLie_t \Upsilon_2 \rrangle$ for $t^a$ the time
    translation Killing field, and
  \item $\llangle \Upsilon_1, \Upsilon_2 \rrangle$ is independent of
    the chosen Cauchy surface $\Sigma$.
  \end{enumerate}
\end{lemma}
Before we prove this lemma, we remark that, e.g. by \eqref{eq:Kinnersley-bilinear}, the bilinear form may be viewed as defined on the initial data of the Teukolsky equation on 
the Cauchy surface $\Sigma$. On an initial data set $\GHPLie_t$ corresponds to the action of a suitably definined Hamiltonian operator $\mathcal{H}$. 
Then item (iii) corresponds to the statement that 
\begin{equation}
    \llangle \Upsilon_1, \mathcal{H} \Upsilon_2 \rrangle = \llangle \mathcal{H} \Upsilon_1, \Upsilon_2 \rrangle, 
\end{equation}
i.e.~to the fact that the Hamiltonian operator is symmetric with respect to our bilinear form. We refer the interested reader to appendix \ref{sec:Lagrangian-Hamiltonian} for details 
on the Hamiltonian formulation of the Teukolsky equation. 

We also note that although we defined our bilinear form on $s=-2$
GHP scalars (i.e., solutions to the adjoint Teukolsky equation), we
could also define a bilinear form on $s=+2$ solutions to the original
Teukolsky equation. In this case, we set
$\llangle\tilde\Upsilon_1, \tilde\Upsilon_2\rrangle \equiv
\Pi_\Sigma[\tilde\Upsilon_1, \Psi_2^{-4/3} \mathcal J
\tilde\Upsilon_2]$. It can be shown that the $s=+2$ bilinear form
satisfies all the same properties as the $s=-2$ form. 
\begin{proof}
  \begin{enumerate}[label=(\roman*),start=1]
  \item This is obvious from the definition.
  \item By explicit calculation, we have with $\pi_{abc} = \epsilon_{abcd} \pi^d$ and $\pi^a$ as in \eqref{pidef},
  \begin{widetext}
    \begin{align}
      \pi_{abc}(\Psi_2^{4/3}\mathcal J \Upsilon_1, \Upsilon_2) &= \epsilon_{dabc} \left[ (\Psi_2^{4/3} \mathcal J \Upsilon_1) (\Theta^d - 4 B^d) \Upsilon_2 - \Upsilon_2 (\Theta^d + 4 B^d) (\Psi_2^{4/3} \mathcal J \Upsilon_1 )\right] \nonumber \\
                                                               &= \mathcal J \epsilon_{dabc} \left[ \Psi_2^{4/3} \Upsilon_1 (\Theta^d - 4 B^{\prime d}) (\mathcal J \Upsilon_2 ) - (\mathcal{J} \Upsilon_2) (\Theta^d + 4 B^{\prime d}) (\Psi_2^{4/3} \Upsilon_1) \right] \nonumber\\
                                                               &= \mathcal J \epsilon_{dabc} \left[ \Upsilon_1 (\Theta^d + 4 B^d) (\Psi_2^{4/3} \mathcal J \Upsilon_2) - (\Psi_2^{4/3} \mathcal J \Upsilon_2) (\Theta^d - 4 B^d) \Upsilon_1 \right] \nonumber\\
                                                               &= - \mathcal J \pi_{abc}(\Psi_2^{4/3}\mathcal J \Upsilon_2, \Upsilon_1),
    \end{align}
    \end{widetext}
    using $\mathcal J^2 = 1$ and~\eqref{eq:gradPsi}. Now integrate over $\Sigma$. Since
    $\mathcal J$ reverses the orientation of $\Sigma$, the claim
    follows.
  \item 
    We first remark that, by Cartan's magic formula, we have that on solutions (where $\pi = \pi_{abc} \dd x^a \wedge \dd x^b \wedge \dd x^c$),
    \begin{equation}
 \mathcal{L}_t \pi = \dd ( t \cdot \pi),
    \end{equation}
    if $\dd \pi=0$. Integrating over $\Sigma$ and using Stokes's theorem,
    \begin{equation}\label{eq:pi-cartan}
    \int_\Sigma  \mathcal{L}_t \pi = \int_{\partial \Sigma} t \cdot \pi = 0.
    \end{equation}
    as, for compact support data, the contribution on $\partial \Sigma$ evaluates to zero. 
    In our case,
    $\Upsilon_1 \in \ker \mathcal{O}^\dagger$, therefore
    $\Psi_2^{4/3} \mathcal J \Upsilon_1 \in \ker \mathcal O$, thus
    $\pi(\Psi_2^{4/3}\mathcal J \Upsilon_1, \Upsilon_2)$ is indeed closed, $\dd \pi=0$.
    On the other hand, we have, since background quantities are all
    GHP-Lie-derived by $t^a = M^{1/3} \xi^a$, and since
    $\mathcal J \GHPLie_t = - \GHPLie_t \mathcal J$, that
    \begin{align}\label{eq:lemmaiiib}
       & \mathcal{L}_t \pi(\Psi_2^{4/3}\mathcal J \Upsilon_1, \Upsilon_2)\nonumber \\
       &\quad= \pi( \Psi_2^{4/3} \GHPLie_t \mathcal J \Upsilon_1, \Upsilon_2) + \pi(\Psi_2^{4/3} \mathcal J \Upsilon_1, \GHPLie_t \Upsilon_2) \nonumber \\
       &\quad= - \pi(\Psi_2^{4/3} \mathcal{J} \GHPLie_t \Upsilon_1, \Upsilon_2) + \pi(\Psi_2^{4/3} \mathcal J \Upsilon_1, \GHPLie_t \Upsilon_2).
    \end{align}
    Inserting this into the left hand side of \eqref{eq:pi-cartan} evaluated on the solutions $\Psi_2^{4/3}\mathcal J \Upsilon_1$ and $\Upsilon_2$ immediately yields the claim.

  \item Holds by Gauss's theorem because $\pi$ is closed on solutions,
    and $\Psi_2^{4/3} \mathcal J$ takes $\ker \mathcal O^\dagger$ into
    $\ker \mathcal O$.
  \end{enumerate}
\end{proof}

We end this section with an explicit expression of our bilinear form in Boyer-Lindquist coordinates and the Kinnersley frame:
\begin{widetext}
\begin{align}\label{eq:Kinnersley-bilinear}
  \llangle \Upsilon_1, \Upsilon_2 \rrangle 
  = 4 M^{4/3} %\Bigg\{
  \int_\Sigma \dd  r \, \dd\theta \dd\phi\, \frac{\sin\theta}{\Delta^2} \Bigg[
   & 
  \Upsilon_1\Big|_{\substack{t\to-t \\ \phi\to-\phi}} \left( \frac{\Lambda}{\Delta}\partial_t + \frac{2Mra}{\Delta}\partial_\phi + 2 \left[ -r - ia\cos\theta + \frac{M}{\Delta}(r^2 - a^2)\right] \right) \Upsilon_2
    \nonumber\\
  & 
    + \Upsilon_2 \left[\left( \frac{\Lambda}{\Delta}\partial_t + \frac{2Mra}{\Delta}\partial_\phi + 2 \left[ -r - ia\cos\theta + \frac{M}{\Delta}(r^2 - a^2)\right] \right) \Upsilon_1\right]_{\substack{t\to-t \\ \phi\to-\phi}}
   \Bigg],
   \end{align}
\end{widetext}   
where we refer to appendix \ref{app:D} for the definitions of 
$\Sigma$, $\Delta$, and  $\Lambda$.\footnote{{Depending on the context and following a standard notation, we use the symbol $\Sigma$ for the metric function or a co-dimension one surface.}} 

\section{Quasinormal mode orthogonality}
\label{sec:Ortho}
\subsection{Quasinormal modes}
Consider modes of the form
\begin{equation}\label{eq:modes}
 {}_s\Upsilon_{\ell m\omega} = e^{-i\omega t + i m \phi} \Rh(r) \Sh(\theta),
\end{equation}
with $m \in \mathbb Z$ and $\omega \in \mathbb C$, in the Kinnersley
frame. This form leads to separation of the spin-$s$
Teukolsky equation~\cite{Teukolsky:1973ha}, $\mathcal O \Upsilon = 0$ (for any integer spin $s$), into an angular equation,
\begin{widetext}
\begin{align}\label{eq:Sph eq}
  \left[\frac{1}{\sin \theta} \frac{ \dd}{\dd \theta}\left(\sin \theta \frac{\dd \,}{\dd \theta} \right) \right. \left. + \left( K -  \frac{m^2+s^2+2 m s \cos \theta}{\sin^2 \theta}  - a^2 \omega^2 \sin^2 \theta -2 a \omega s \cos \theta \right) \right] \Sh(\theta) = 0,
\end{align}
and a radial equation,
\begin{align}\label{eq:radial}
   \left[ \Delta^{-s} \frac{\dd}{\dd r} \left( \Delta^{s+1} \frac{\dd}{\dd r} \right) \right. \left. + \left( \frac{H^2 - 2 i s (r-M)H}{\Delta} + 4 i s \omega r+2 a m \omega - K +s(s+1) \right) \right] \Rh(r) = 0,
\end{align}
\end{widetext}
with $H \equiv (r^2+a^2)\omega - a m$. Here $K$ is a separation
constant. Imposing regularity at the poles $\theta=0,\pi$, the angular
equation leads to a discrete set of modes $\Sh$ and separation
constants $\Kh$, both of which are indexed by
$\ell \in \mathbb Z^{\ge \max(|m|, |s|)}$. The functions
$\Sh(\theta)e^{im\phi} $ are known as spin-weighted spheroidal
harmonics~\cite{Teukolsky:1973ha}. For $\omega \in \mathbb R$, the
angular problem reduces to a Sturm-Liouville eigenvalue problem. 
Modes with the same $s$,
$m$, and real $\omega$, but different $\ell$ are orthogonal, and
we normalize them such that
\begin{equation}\label{eq:theta-orthogonality}
  \int_0^\pi \dd\theta\, \sin\theta \, {}_sS_{\ell m\omega}(\theta) {}_sS_{\ell'm\omega}(\theta) = \delta_{\ell\ell'}.
\end{equation}
Orthogonality can be checked by verifying that the angular operator is
symmetric with respect to this product. 

To discuss boundary conditions of the radial equation it is convenient to
introduce a ``tortoise'' coordinate $\dd r_*= (r^2+a^2)/\Delta \dd r$, see
\eqref{eq:rstar}. 
For fixed $s,l,m,\omega$ one considers the  solutions
$R^{\rm in}$ and $R^{\rm up}$ ``defined'' by the ``boundary conditions''
\begin{subequations}\label{eq:R bcs}
  \begin{align}
    &R^{\rm in}  \sim  \frac{e^{-ikr_*}}{\Delta^{s}}, \qquad  r_*\to-\infty,\\
    &R^{\rm up}  \sim  \frac{e^{i\omega r_*}}{r^{2s+1}}, \qquad  r_*\to\infty,
  \end{align}
\end{subequations}
where 
$k \equiv \omega-m\Omega_H$, where $\Omega_H$ is the angular frequency of the outer horizon
$\Omega_H=a/(2Mr_+)$, and where the radii of the inner- and outer horizons (roots of $\Delta$) are denoted by $r_\pm$, respectively. 

The conditions  \eqref{eq:R bcs} correspond physically to the
absence of incoming radiation from the past horizon and past null
infinity, respectively. As stated \eqref{eq:R bcs} do not really pick out uniquely a solution in the case ${\rm Im} \omega<0$ because we may always add a multiple of the subdominant solution as $|r_*| \to \infty$ without affecting the asymptotic behavior. 
More precisely, mode solutions may be obtained via series expansions \cite{Leaver1986}, involving three-term recurrence relations for the coefficients.
 Selecting the so-called ``minimal solution'' \cite{gautschi1967computational}  of the recurrence relations ensures that the series represenation converges at the horizon (in) or infinity (up).\footnote{This definition is satisfied by a radial solution of the form
\begin{equation}
    R(r) = e^{i\omega r}(r-r_-)^{-1-s+i\omega+i\sigma_+}(r-r_+)^{-s-i\sigma_+} f(r) ,
\end{equation}
where $\sigma_+=(\omega r_+-am)/(r_+-r_-)$ and $f(r)=\sum_{n=0}^{\infty}d_n\left(\frac{r-r_+}{r-r_-}\right)^n$ with $d_n$ coefficients that are a minimal solution to a three-term recursion relation~\cite{Leaver1985} so that the series is uniformly absolutely convergent as $r\rightarrow\infty$}
 Imposing both of these conditions simultaneously
\footnote{The problem is made
complicated, however, because $\omega$ and $K$ appear in both the
angular and radial equations, $\omega$ nonlinearly. One must jointly
solve both equations to obtain a self-consistent solution of this
nonlinear eigenvalue problem. Using Hamiltonian methods (see appendix \ref{sec:Lagrangian-Hamiltonian})
one can recast this as the eigenvalue problem $\mathcal{H} \Upsilon = i\omega \Upsilon$,
i.e., the problem is linear in $\omega$, but the angular and radial
problems remain coupled.}
gives rise to a discrete set of quasinormal modes $\omega_n \in \mathbb{C}$, where $n=0,1,2,\ldots$ are the so-called ``overtone'' numbers.   We restrict to frequencies with ${\rm Im}\, \omega \le 0$, as modes growing exponentially in time are not in the specturm of Kerr \cite{Whiting:1988vc}.

%%%%%%%%%%%%%%%%%%%%%%%%%%%%%%%%%%%%%%%%%%%%%%%%%%%%%%%%
\begin{figure*}
\centering
\includegraphics[trim={0.cm 4.cm 0.cm 2.cm},clip,width=0.49\linewidth]{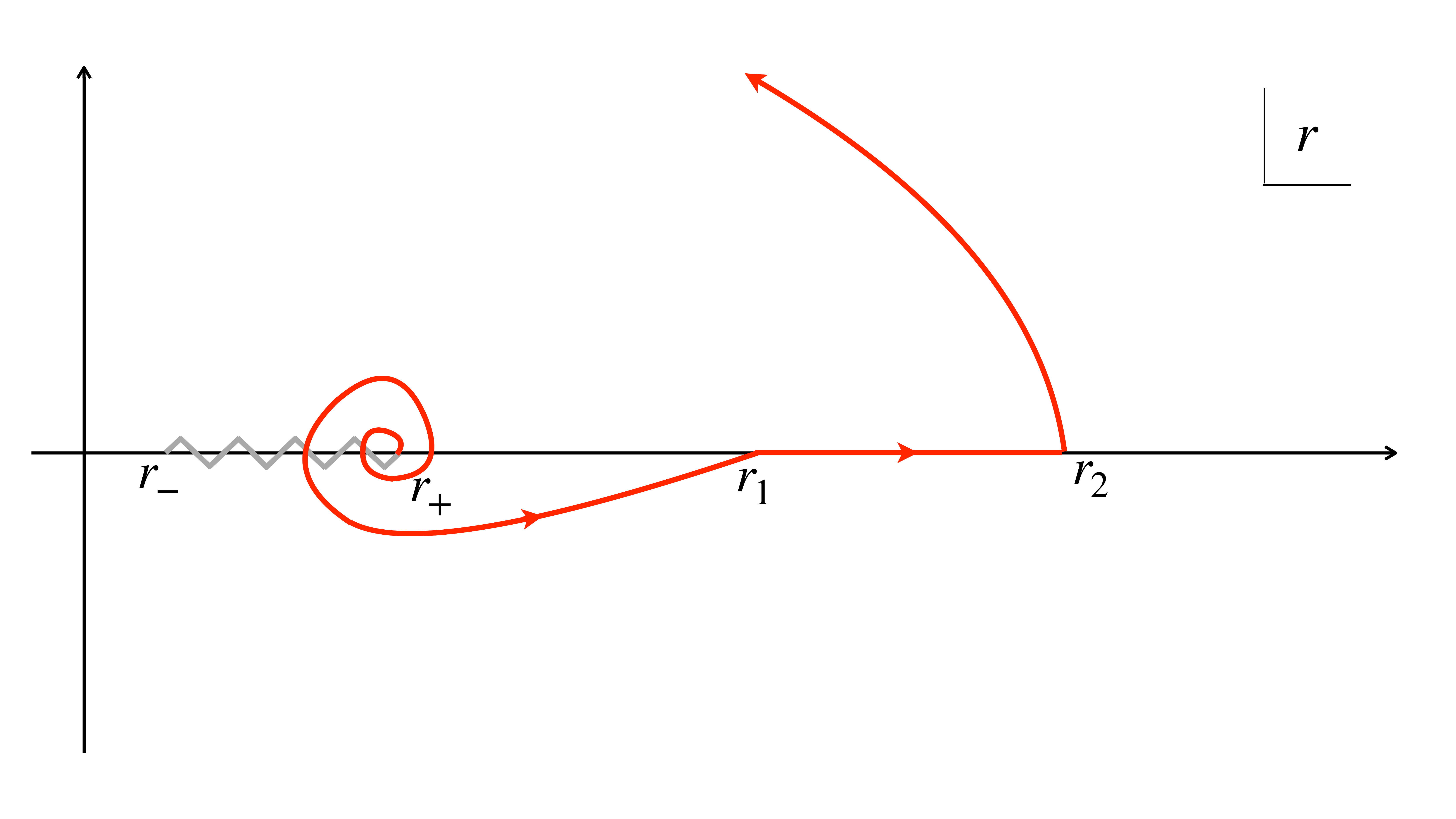}
\includegraphics[trim={0.cm 4.cm 0.cm 2.cm},clip,width=0.49\linewidth]{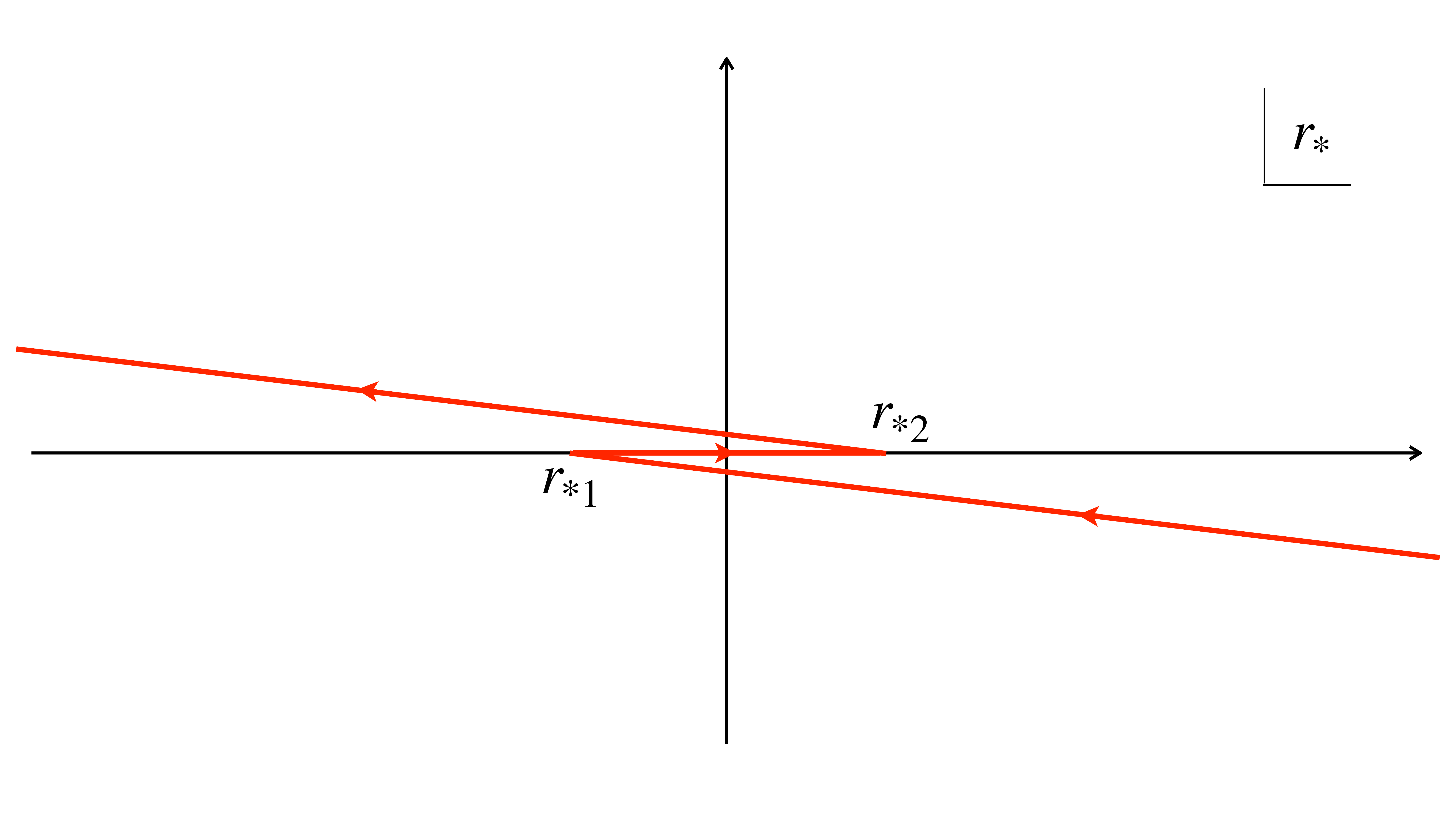}
\caption{{\it Left:} Sketch of the complex $r$ contour $C_*$ defining the bilinear form on quasinormal modes. The contour cannot be pulled back to the real axis because the integrand crosses (an infinite number of) different sheets associated with the branch points $r_-$ and $r_+$. {\it Right:} Same contour, but in the complex $r_*$ plane.
Note that this contour cannot be pulled back to the real axis due to the presence of Stokes lines along which the integrand of the bilinear form would diverge.}
\label{fig:r_contour}
\end{figure*}
%%%%%%%%%%%%%%%%%%%%%%%%%%%%%%%%%%%%%%%%%%%%%%%%%%%%%%%%

\subsection{Bilinear form on quasinormal modes}

We would now like to extend our definition of the bilinear form $\llangle \cdot , \cdot \rrangle$, originally 
only for compactly supported solutions/data on the Cauchy surface $\Sigma$, to quasinormal modes. 
The immediate problem is that, according to the boundary conditions on the corresponding solutions 
to the radial equation, these blow up both at the horizon $r=r_+$ and infinity $r \to \infty$. 
In this subsection, inspired by the work of~\cite{LeungModes94}, we show that the Kerr
bilinear form can be defined for quasinormal mode data by a suitable
deformation of the radial integration into the complex plane.\footnote{
In the quantum mechanics literature, this method is known also as (exterior)
complex scaling \cite{Aguilar:1971ve}. 
Complex scaling and complex integration contours have already been used in the context of black hole quasinormal modes, see for instance~\cite{Bony2007,Dyatlov:2011jd} and~\cite{Glampedakis:2003dn,Leaver1986b}. } 

Consider the bilinear form acting on two quasinormal modes with
quasinormal frequencies $\omega_1$ and $\omega_2$. The 
integrand in the bilinear
form~\eqref{eq:mode-bilinear} goes as
$\sim e^{\pm i(\omega_1+\omega_2)r_\ast}$ as $r_\ast \to \pm\infty$,
and therefore diverges exponentially for
$\Im(\omega_1 + \omega_2) < 0$, which is the case for all modes that decay
in time. Therefore, we clearly see that the bilinear form as defined for compact support data~\eqref{eq:Kinnersley-bilinear} is divergent.

We can obtain a finite bilinear form by analytic continuation in $r$.
The radial mode functions $R^{\rm in/up}(r)$ are analytic with branch points at $r=r_\pm$ \cite{Leaver1986}, and we take the branch cut as the wiggly line in Fig.~\ref{fig:r_contour} going from $r_+$ to $r_-$. We take the branch cut for the tortoise coordinate
\eqref{eq:rstar} $r_*(r)$ to be identical, so that we can think of both the radial functions $R^{\rm in/up}$ and $r_*$ as 
defined on the same multisheeted covering of the twice cut complex $r$-plane. The integrand of the bilinear form, given by the 3-form 
$\pi_{abc} = \epsilon_{abcd} \pi^d$ [see \eqref{pidef}] evaluated on two mode 
functions as in \eqref{bilinear} or equivalently \eqref{eq:Kinnersley-bilinear}, therefore has an analytic continuation on the multi-sheeted complex $r$-plane. 

In \eqref{bilinear} or equivalently \eqref{eq:Kinnersley-bilinear}, we now define an integration contour going into 
this complex $r$-plane as shown qualitatively 
in fig. \ref{fig:r_contour}. In terms of $r_*(r)$, 
which is a function on the same multi-sheeted complex $r$-plane, the contour is defined in such a way that $0 < \arg((\omega_1 + \omega_2) r_\ast) < \pi$
on the right limit, and
$-\pi < \arg((\omega_1 + \omega_2) r_\ast) < 0$ on the left, then as
$|r_\ast| \to \infty$, 
the volume integral will converge exponentially with $|r_\ast|$.

To achieve this for any $\Im(\omega_1 + \omega_2) < 0$ me may take a snake
shaped contour $u \mapsto r_*(u,\epsilon)$ of the radial
coordinate in the complex $r_*$ plane, with the properties
\begin{equation}\label{eq:contour}
 \begin{cases}
   r_*(u,\epsilon)=u
   &\text{for $r_{*1} < r_* <r_{*2}$}\\
   \arg r_*(u,\epsilon) \to +\pi - \epsilon  & \text{for $r_* \to \infty$}\\
   \arg r_*(u,\epsilon) \to 0 + \epsilon &\text{for $r_* \to -\infty$,}
 \end{cases}
\end{equation}
{where $r_{*1}<0$, $r_{*2}>0$} can in principle be chosen arbitrarily. We give a sketch of this contour, $C_*$, which corresponds 
to one in terms of $r$, in the right panel of Fig.~\ref{fig:r_contour}. The corresponding 3-dimensional
 submanifold (depending on $\epsilon > 0$ and on $t \in \R$) of the analytically continued Kerr manifold $M_{\CC}$ is denoted by
$\Sigma_{\CC} = \{ (t,r_*(u,\epsilon),\theta,\phi)  \mid u \in \R \}$.
In practice, the angle $\epsilon>0$ is chosen sufficiently small such that the integral in the following definition of the bilinear form converges, 
\begin{align}
  \llangle \Upsilon_1, \Upsilon_2 \rrangle &= \Pi_{\Sigma_{\CC}}[\Psi_2^{4/3}\mathcal J \Upsilon_1, \Upsilon_2].  \\
 \nonumber
\end{align}
Replacing $\Sigma$ with the contour $\Sigma_{\mathbb{C}}$ as described in section~\ref{sec:bilinear_tphi}, 
thanks to the analyticity of the integrand and its fall off on $\partial \Sigma_{\mathbb{C}}$, all properties of the bilinear form of of lemma \ref{lemma:compactsupport} continue to hold on quasinormal modes. In particular, from item (iii) of lemma \ref{lemma:compactsupport}, we get $(\omega_1-\omega_2) \llangle \Upsilon_1 , \Upsilon_2 \rrangle=0$ for a pair of quasinormal modes with complex frequencies $\omega_1, \omega_2$. Furthermore, by (iv), the value of the bilinear form is independent of the precise choice of $t$, details of the complex integration contour such as the asymptotic angle $\epsilon$ against the real half-axes and/or  $r_{*1}, r_{*2}$, as long as the integrand is exponentially decaying. 

\begin{corollary}[Orthogonality of quasinormal modes]
  Let $\Upsilon_1$ and $\Upsilon_2$ be quasinormal modes for the $s=2$
  Teukolsky equation
  with frequencies $\omega_1$ and $\omega_2$. Then either
  $\llangle \Upsilon_1, \Upsilon_2 \rrangle = 0$ or
  $\omega_1 = \omega_2$.
\end{corollary}

Our bilinear form takes the following form on quasinormal mode solutions~\eqref{eq:modes}. After plugging two $s=-2$ mode solutions in separated form
into~\eqref{eq:Kinnersley-bilinear}, we can carry out the $\phi$
integration to obtain 
\begin{widetext}
\begin{eqnarray}\label{eq:mode-bilinear}
  &&\llangle \Upsilon_{\ell_1m_1\omega_1}, \Upsilon_{\ell_2m_2\omega_2} \rrangle \nonumber\\&= &8\pi M^{4/3} \delta_{m_1m_2} e^{-i(\omega_2-\omega_1)t}
  \int_{C_*} \dd r_* \int_0^\pi  \dd\theta\, \frac{\sin\theta}{(r^2+a^2) \Delta} S_1(\theta) S_2(\theta)  R_1(r)  R_2(r)  %\nonumber\\
 % & \qquad \qquad
  \bigg( - \frac{i\Lambda}{\Delta}(\omega_1+\omega_2)   \nonumber\\
 &&\qquad \qquad \qquad  \qquad \qquad \qquad \qquad \qquad \qquad + \frac{2iMra}{\Delta}(m_1+m_2) + 4 \left[ -r - ia\cos\theta + \frac{M}{\Delta}(r^2 - a^2)\right] \bigg)
\end{eqnarray}
\end{widetext}
with $C_*$ the contour for the $r_*$-integration described above and the Kerr quantities $\Delta, \Sigma, \Lambda$ as given in 
appendix \ref{app:D}.

The integrands depend on $\theta$ and $r$ in a nonfactorizable way, so
this expression is the best that can be achieved in general: for Kerr,
the orthogonality relation expressed by the previous corollary 
(vanishing of the above inner product for $\omega_1 \neq \omega_2$) 
is fundamentally two-dimensional. This has
to do with the fact that the orthogonality
relation~\eqref{eq:theta-orthogonality} for spin-weighted spheroidal
harmonics occurs between modes of different $\ell$ but the \emph{same}
$m$ and $\omega$; if $\omega_1 \ne \omega_2$, then no such relation
exists, and one cannot expect to be able to perform the $\theta$
integration to obtain a $\delta_{\ell_1\ell_2}$ factor. 

In the $a\to0$ Schwarzschild limit, however, the integral \emph{does}
factorize: the $\theta$ dependence of the integrand reduces to the
$\sin\theta$ volume factor on the sphere, the spheroidal harmonics
reduce to spherical harmonics (independent of $\omega$), and the
$\theta$ integral is proportional to $\delta_{\ell_1\ell_2}$. One is left
with a radial integration, which must vanish for
$\omega_1 \ne \omega_2$. 

{In Fig.~\ref{fig:ortho}, we numerically evaluate the bilinar form on two pairs of Kerr quasinormal modes. We do so along the most convergent contour -- along which the integrand is purely damped -- for the given pair of frequencies $\omega_1$, $\omega_2$: $\arg r_*(u) + \arg (\omega_1+\omega_2) =\pi/2$ for $r_* \to \infty$ and $\arg r_*(u) + \arg (\omega_1+\omega_2) = - \pi/2$ for $r_* \to -\infty$. In other words, the two complex sections of the contour are given by $r_*(u)= r_{*,\rm 1,2}+ u e^{-i \arg(\omega_1+\omega_2)+i\pi/2}$ with $u<0$ in the section emanating from $r_{*1}$ and with $u>0$ in the other (see again Fig.~\ref{fig:r_contour}). We then integrate $u$ between $0$ and a finite $u_\text{lower, upper}$, respectively, and study the convergence of the integral as we take $u_\text{upper, lower} \rightarrow\pm\infty$.
As we can see from the figure,} the contour integral in the bilinear form converges quite well, 
which is useful in practice when using it to extract excitation coefficients, as we describe in the next section. 
Furthermore, since orthogonality is an exact result for quasinormal modes, it can be used potentially as a benchmark check 
for approximations. For example, we have considered approximations to quasinormal modes based on a matched asymptotic 
expansion for near-extremal black holes, and have found that the orthogonality relation is typically satisfied to a very high accuracy. 

%%%%%%%%%%%%%%%%%%%%%%%%%%%%%%%%%%%%%%%%%%%%%%%%%%%%%%%%
\begin{figure}
\centering
\includegraphics[width=1.\linewidth,trim={.cm .1cm .1cm .1cm},clip]{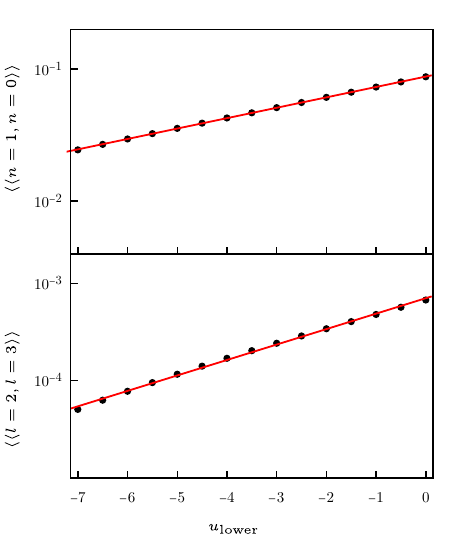}
\caption{Numerical check of the orthogonality between two Kerr quasinormal modes with the same $l=m=2$ and different $n=0$, 1 (upper panel) and modes with the same $n=0$, $m=2$ and different $l=2$, 3 (lower panel). We show the result of the numerical evaluation of the bilinear form~\eqref{eq:mode-bilinear} along the most convergent contour (black points), integrating {the complex sections of the contour up to a finite $u_\text{lower, upper}$.} Because of the {difficulty in handling the} branch cut {in Mathematica}, the lower integration limit {$u_\text{lower}$ sets the overall accuracy of this numerical evaluation of the bilinear form. We also show an exponential fit converging to zero as $u_\text{lower}\to-\infty$ (red line).}  We use the mode solutions {(normalized at the horizon)} provided by the Black Hole Perturbation Toolkit~\cite{BHPToolkit}. In this example, we set $M=1$, $a=0.7$, $r_{*\rm 1}=-8$, and $r_{*\rm 2}=5$.
}
\label{fig:ortho}
\end{figure}

We remark that
our ``norm'' on quasinormal modes 
has some similarities with the ``norm'' of resonant state 
wave functions in quantum mechanics defined by~\cite{Zeldovich:1961theory}. Rather than taking the integral of $|\psi|^2$, the  ``norm''
used by \cite{Zeldovich:1961theory} also involves $\psi^2$, whereas our bilinear is complex linear in both arguments as opposed to an inner product (anti-linear in the first argument, complex linear in the second).
Our regularization procedure differs from that proposed by~\cite{Zeldovich:1961theory} but was rather inspired by the investigations \cite{LeungModes94} in the context of leaky optical
one-dimensional cavities, and on Schwarzschild black holes in~\cite{Ching:1993gt}.
In \cite{leung1997twoa,leung1997twob} it was recognized that phase space was the natural setting for the bilinear form.\footnote{A
variational method for computing quasinormal frequencies of ``dirty''
Schwarzschild black holes was developed in~\cite{Leung:1999rh,Leung:1999iq}.}
In fact, in several ways our work was inspired by some of these papers: we work
within the Teukolsky formalism, we arrive at the bilinear form
starting from the symplectic form, and we recognize the fundamental
importance of the $t$--$\phi$ reflection symmetry. 

\medskip
%%%%%%%%%%%%%%%%%%%%%%%%%%%%%%%%%%%%%%%%%%%%%%%%%%%%%%%%

\section{Excitation coefficients}\label{sec:Lap transform} 

If $\llangle \cdot , \cdot \rrangle$ were an honest to God scalar product in a Hilbert space and $\{ {}_s\Upsilon_{\ell mn} \}$ an orthonormal basis, then an arbitrary 
wave function $\Upsilon_s$ could evidently be expanded as
\begin{align}
\label{excited}
\Upsilon_s = \sum_{\ell mn} c_{\ell mn} \, {}_s\Upsilon_{\ell mn},
\end{align}
where the excitation coefficients are
\begin{align}
c_{\ell mn} =\frac{\llangle {}_s\Upsilon_{\ell mn} , \Upsilon_s \rrangle }{\llangle {}_s\Upsilon_{\ell mn} , {}_s\Upsilon_{\ell mn} \rrangle }. \label{eq:excitation coeff}
\end{align}
Here $\sum_{\ell m n}$ denotes  $\sum_{\ell = \vert s\vert}^\infty\sum_{m=-\ell}^\ell \sum_{n=0}^\infty$. In the present context, $\llangle \cdot,\cdot\rrangle$ is of course 
only a symmetric bilinear form on solutions to the spin $s$ Teukolsky equation (for the case of interest in this paper, $s=-2$). It is neither positive definite, nor is the set of quasi-normal modes, 
while being orthogonal, in any obvious mathematical sense a complete basis for a reasonable function space in as far as we can see. 

Inspired by~\cite{LeungModes94}, we will nevertheless show in this section 
that for solutions $\Upsilon_s$ to the adjoint Teukolsky equation with compact support on a Cauchy surface $\Sigma$, 
the above expansion can formally be ``derived'' in the Laplace transform formalism \cite{Leaver1986b,Nollert:1999ji} for the retarded propagator if we deform the frequency 
integration contours into the complex plane and collect only contributions from the quasinormal mode frequencies. Thus, \eqref{eq:excitation coeff}, while not an exact 
equality, is expected to capture the transient behavior of the solution $\Upsilon_s$.

\subsection{Laplace transform}\label{eq:Laplace}

The Laplace transform $\hat f(\omega) = L f(t)$ of a function $f(t)$ is given by
\begin{equation}\label{eq:Lap def}
\hat f(\omega) = \int_0^\infty e^{i\omega t} f(t) \dd t,
\end{equation}
 where $\Im \omega >0$. The Laplace transform is related to the Fourier transform $\mathcal F \!f = \int_{-\infty}^{\infty}  e^{i \omega t} f(t)  \dd t$  by sending $f(t) \to  f(t) \theta (t)$, where $\theta(t)$ is the Heaviside distribution. Sufficient conditions for the existence of $\hat f(\omega)$ are that the function $f(t)$ be Riemann integrable (continuous except on sets of measure zero) on every closed sub-interval of the path of integration and that it be of exponential order; i.e. at any  $t$ one can find constants $a$ and $N$ such that $\vert e^{-at} f(t) \vert < N$.  If the Laplace integral exists for some value of $\omega=\omega_0$, then it also exists for all $\omega$ with $ \Im \omega>\Im \omega_0$.  The lowermost $\Im \omega_0$ where convergence occurs is called the abscissa of convergence and the region above this line called the convergence region. The function $\hat f(\omega)$ is analytic in the convergence region.  

The Laplace transform formalism is naturally adapted to the study of causal dynamics of linear second-order systems,  as it incorporates the initial data into a source by taking time derivatives into field values at the initial time
\begin{align}\label{eq:lap derivs}
Lf'(t) = &-i\omega \hat f(\omega)-f(0), \\
Lf''(t) = &-\omega^2 \hat f(\omega)+i\omega f(0) - f'(0).
\end{align}
The Laplace transform $\tUs$ of the spin-$s$ master function $\Us$ is given by
\begin{equation}\label{eq:tilde U}
\tUs(\omega,r,\theta,\phi) = \int_{0}^{\infty} e^{i \omega t} \Us(t,r,\theta,\phi) \dd t.
\end{equation}
This decomposed into modes in the usual way
\begin{equation}\label{eq:tilde Us mode}
\tUs = 
\sum_{\ell m} \Sh(\theta) \, \Rh(r) e^{i m \phi}.
\end{equation}
The inverse transform is given by
\begin{equation}\label{eq:inverseLap}
\Upsilon_s(t,r,\theta,\phi)= \frac{1}{2\pi}\int_{-\infty+ic}^{\infty+ic} e^{-i\omega t} \tUs(\omega,r,\theta,\phi)\, \dd \omega,
\end{equation}
where $c>0$ is chosen such that the integral contour lies within the convergence region.

To formulate the initial data problem within the mode decomposition, we  take the Laplace transform of the Teukolsky master equation 
and substitute \eqref{eq:tilde Us mode}. We then collect the terms in the  master equation with transformed time derivatives to the right-hand side and project onto the angular mode function. 
This yields  a sourced equation for the radial function
\begin{equation}\label{eq:radeq formal}
\mathcal{L} \Rh = \Ih.
\end{equation}
Here, $\mathcal{L}$ is given in \eqref{eq:radial} and the source $\Ih=\Ih(r)$ is comprised of  $(\ell,m)$-projected initial data~\cite{Campanelli:1997un}:
\begin{widetext}
\begin{align}\label{eq:I source modes}
 \Ih= \int_0^{2\pi}\! \! \int_0^\pi \bigg[ 
\frac{\Lambda}{\Delta} \left( \pd_t \Us - i \omega \Us \right) -2s \Big( \frac{M(r^2-a^2)}{\Delta} -r -ia\cos\theta \Big) \Us + \frac{4 M ar}{\Delta} \pd_\phi \Us\bigg]_{t=0}\Sh(\theta)\, e^{-i m \phi} \sin \theta\, \dd \theta \, \dd \phi
\end{align}
\end{widetext}
which we take to be of compact support.
 Imposing the outgoing boundary conditions \eqref{eq:R bcs} fixes the freedom of homogeneous solutions to \eqref{eq:radeq formal} and therefore defines a radial Green's function
%\begin{equation}\label{eq:tf def}
${}_s  g_{\ell m\omega}(r,r') =  \Rh^{\rm in}(r_<)\Rh^{\rm up}(r_>)/\mathcal W $
where $r_< \,\,(r_>)$ is the lesser (greater) of $r$ and $r'$. Here, for any two solutions of the radial equation at fixed
$s, m, \ell, \omega$, the (``$\Delta$-scaled'') Wronskian
\begin{equation}\label{eq:Wronskian}
\mathcal W[R_1,R_2] = \Delta^{1+s}\left[ R_1 \frac{\dd R_2}{\dd r} - R_2 \frac{\dd R_1}{\dd r} \right]
\end{equation}
has been defined which is independent of $r$. If $R_1$ and $R_2$ are linearly dependent, then
the Wronskian vanishes. Thus, if we take $R_1 \to R^{\text{in}}$,
$R_2 \to R^{\text{up}}$, the Wronskian vanishes when $\omega$ attains
a quasinormal frequency.

The quasinormal mode contribution to $\Upsilon_s$ can be found by closing the contour of the Laplace integral in the lower-half complex $\omega$ plane, and can be expressed 
as a  discrete sum over the residues of the radial Green's function arising at the points in the complex frequency plane where the Wronskian vanishes:
%\begin{widetext}
\begin{align}\label{eq:oops n}
\Upsilon_{s} =& -i \sum_{n\ell m} e^{-i\omega_n t+im\phi} {}_sS_{\ell m n}(\theta) \nonumber\\
&\times \int_{r_+}^{\infty} \frac{
\Rq{n}^{\rm in}(r_<) \Rq{n}^{\rm up}(r_>)}{\dd\mathcal W/ \dd\omega\vert_{\omega_n}} {}_s I_{\ell m n}(r')\Delta^s(r') \dd r',
\end{align}
%\end{widetext} 
where ${}_s I_{\ell m n}={}_s I_{\ell m \omega}\vert_{\omega=\omega_n}$.  In considering only the poles when closing the contour, we are effectively ignoring 
the early-time ``direct'' contribution from the large-$\omega$ arc and the late-time ``tail'' contribution resulting from the branch point at zero frequency.
Thus, the $=$ sign in the above equation is not actually justified and should be understood as meaning this approximation.

On a quasinormal mode,
 $R^{\rm in}$ is a constant multiple of $R^{\rm up}$, and either may be moved outside the radial integral to write the field as
\begin{equation}\label{eq:oops nice}
\Upsilon_{s} = \sum_{n \ell m} c_{\ell mn} \,{}_s\Upsilon_{lmn},
\end{equation}
where we have isolated the familiar form of the excitation coefficient
\begin{align}\label{eq:qnm exc}
c_{n\ell m} &=   -\frac{i}{\dd \mathcal W/\dd\omega\vert_{\omega_n}}\int_{r_+}^{\infty}
 {}_s I_{\ell m n}(r') \, {}_s R_{\ell mn}(r') \Delta^s(r') \, \dd r'.
\end{align}

\subsection{Equivalence between  
\eqref{eq:qnm exc} and \eqref{eq:excitation coeff}}  

 To begin, for a given
radial function $R$, and $\ell, m, \omega$, we can define a $s=-2$ GHP
scalar $\Upsilon_{\ell m \omega}$ in separated form~\eqref{eq:modes}
by appending a spin-weighted spheroidal harmonic and $e^{-i\omega t}$
time-dependence. For the time being, we do not require $R$ to satisfy
any equation. We have the following lemma relating the Wronskian to
$t\cdot\pi$ integrated over the 2-sphere.

\begin{lemma}\label{lemma:wronskian-boundary}
  Let $\Upsilon_1, \Upsilon_2 \GHPwt (-4,0)$ be two GHP scalars in
  separated form~\eqref{eq:modes}, with the same $m, \ell, \omega$,
  where $S_1$, $S_2$ are normalized spin-weighted spheroidal harmonics
  solving the angular equation, but where $R_1, R_2$ are not
  necessarily solutions to the radial equation. Then
  \begin{equation}
    8\pi M^{4/3} \mathcal{W}[R_1,R_2] = \int_{S^2(t,r)} t \cdot \pi(\Psi_2^{4/3}\mathcal J \Upsilon_1, \Upsilon_2),
  \end{equation}
  where $S^2(t,r)$ is a sphere of constant $t$ and $r$, $\pi_{abc} = \epsilon_{abcd} \pi^d$ and $\pi^a$ as in \eqref{pidef}.
\end{lemma}
\begin{proof}
Consider the Cauchy surface $\Sigma(t) = \{t={\rm const.}\}$ in Boyer-Lindquist coordinates. The future directed normal 
to $\Sigma$ and induced area element on $S^2(t,r)$ are given by, respectively
\begin{align}
  \nu^a =& \left( \sqrt{\frac{\Lambda}{\Delta\Sigma}}, 0, 0, \frac{2Mar}{\sqrt{\Delta\Sigma\Lambda}}\right), \\
  \dd A =& \sqrt{\Sigma\Lambda} \sin\theta \, \dd \theta \dd \phi. 
\end{align}
From the first relation on can read off the lapse function $N$ of $t^a$ from $\nu^a = (t^a - N^a)/N$. 
The action of the reflection reverses $\nu^a$ and from this fact and the formula for $\pi^a$, see \eqref{pidef}, one can deduce that 
\begin{widetext}
  \begin{align}\label{eq:int-tdotpi}
  \int_{S^2(t,r)} t \cdot \pi(\Psi_2^{4/3}\mathcal J \Upsilon_1, \Upsilon_2) =  \int_{S^2(t,r)} N \Psi_2^{4/3}\left\{ (\mathcal J\Upsilon_1) r^a(\Theta_a - 4 B_a) \Upsilon_2 - \Upsilon_2 \mathcal J [r^a (\Theta_a - 4 B_a) \Upsilon_1 ]\right\} \dd A
 \end{align}
\end{widetext}
where $r^a$ is the normal to $S^2(r,t)$ inside $\Sigma(t)$.
An explicit calculation shows that in the Kinnersley frame,
  \begin{equation}
    r^a(\Theta_a - 4B_a)\Upsilon = \sqrt{\frac{\Delta}{\Sigma}}\partial_r\Upsilon - 2\frac{(r-M)}{\sqrt{\Delta\Sigma}}\Upsilon.
  \end{equation}
  Using this, as well as expressions for $N$, $\dd A$, and
  $\mathcal J$ [using~\eqref{eq:boostParams}], we obtain
 \begin{widetext}
  \begin{equation}
    \int_{S^2(t,r)} t \cdot \pi(\Psi_2^{4/3}\mathcal J \Upsilon_1, \Upsilon_2)
    = \frac{4M^{4/3}}{\Delta(r)} \left( R_1\frac{\dd R_2}{\dd r} - R_2 \frac{\dd R_1}{\dd r}\right) \int_0^{\pi}\int_0^{2\pi} \ S_1(\theta) S_2(\theta) \, \sin\theta  \dd \theta \dd \phi .
  \end{equation}
\end{widetext}
Finally, perfoming the integration and using the normalization~\eqref{eq:theta-orthogonality} for the angular
  functions, we obtain the result.
\end{proof}

Next, we take $R_1$ and $R_2$ to be solutions ingoing at the horizon
and outgoing at infinity. Considered as a function of $\omega$, the
Wronskian vanishes at quasinormal frequencies $\omega_n$, because at
these frequencies the two solutions become linearly dependent. The
first derivative with respect to $\omega$, however, is proportional to
the ``norm'' of the quasinormal mode.
\begin{lemma}\label{lemma:Wronskian-derivative}
  Let $R_\omega^{\mathrm{in}}, R_{\omega}^{\mathrm{up}}$ be solutions to
  the radial equation for fixed $s=-2,\ell,m$, and allowing $\omega$
  to vary, that are ingoing at the horizon and outgoing at infinity,
  respectively, as in~\eqref{eq:R bcs}. Construct
  $\Upsilon^{\mathrm{in}}_\omega, \Upsilon^{\mathrm{up}}_\omega \GHPwt
  (-4,0)$ as mode solutions based on the radial functions. Then the
  derivative of the Wronskian at a quasinormal frequency $\omega_n$
  can be written
  \begin{equation}\label{eq:excitation_factor}
    \left. \frac{\dd}{\dd \omega}\mathcal{W}[R^{\mathrm{in}}_\omega, R^{\mathrm{up}}_\omega] \right|_{\omega = \omega_n} = \frac{-i}{8\pi M^{4/3}} \llangle \Upsilon^{\mathrm{in}}_{\omega_n}, \Upsilon^{\mathrm{up}}_{\omega_n} \rrangle.
  \end{equation}
\end{lemma}
\begin{proof}
  Consider the current
  $\pi\left(\Psi_2^{4/3} \mathcal J \Upsilon^{\text{in}}_{\omega_n},
    \Upsilon^{\text{up}}_\omega\right)$ evaluated at a generic
  frequency $\omega$ and a quasinormal frequency $\omega_n$.  By
  Cartan's magic formula and the fact that $\pi$ is closed on
  solutions,
  \begin{align}
    &\dd\left(t \cdot \pi\left(\Psi_2^{4/3} \mathcal J \Upsilon^{\text{in}}_{\omega_n}, \Upsilon^{\text{up}}_\omega\right)\right) \nonumber \\
    &\quad =  \Lie_t \pi\left(\Psi_2^{4/3} \mathcal J \Upsilon^{\text{in}}_{\omega_n}, \Upsilon^{\text{up}}_\omega\right) 
    \nonumber\\
    &\quad =  -i(\omega - \omega_n)\pi\left(\Psi_2^{4/3} \mathcal J \Upsilon^{\text{in}}_{\omega_n}, \Upsilon^{\text{up}}_\omega\right) ,
  \end{align}
  where in contrast to the previous lemma the right side does not
  vanish on account of the different frequencies. We integrate over a partial Cauchy surface $S$, and
  apply Stokes's theorem,
  \begin{align}\label{eq:balance}
    &\int_{\partial S} t \cdot \pi(\Psi_2^{4/3} \mathcal J \Upsilon^{\text{in}}_{\omega_n}, \Upsilon^{\text{up}}_\omega)\nonumber\\
     &\quad= -i (\omega - \omega_n) \int_S \pi(\Psi_2^{4/3} \mathcal J \Upsilon^{\text{in}}_{\omega_n}, \Upsilon^{\text{up}}_\omega).
  \end{align}
  Next, we
  differentiate this equation with respect to $\omega$ and take the
  limit $\omega\to \omega_n$. On the right side, we trivially get
  \begin{equation}
    \left.\frac{\dd}{\dd\omega}\right|_{\omega=\omega_n} \text{r.h.s. of \eqref{eq:balance}} 
    = -i \int_S \pi(\Psi_2^{4/3} \mathcal J \Upsilon^{\text{in}}_{\omega_n}, \Upsilon^{\text{up}}_{\omega_n}).
  \end{equation}
  The left side can be expressed as three terms, namely 
  \begin{align}
    &\left.\frac{\dd}{\dd\omega}\right|_{\omega=\omega_n} \text{l.h.s. of \eqref{eq:balance}}  \nonumber\\
     &\quad= \int_{\partial S_+} t\cdot\pi\left(\Psi_2^{4/3} \mathcal J \Upsilon^{\text{in}}_{\omega_n}, \left.\frac{\dd}{\dd\omega}\right|_{\omega=\omega_n} \Upsilon^{\text{up}}_\omega\right) \nonumber\\
    &\qquad- \left.\frac{\dd}{\dd\omega}\right|_{\omega=\omega_n} \int_{\partial S_-} t\cdot\pi\left(\Psi_2^{4/3} \mathcal J \Upsilon^{\text{in}}_\omega, \Upsilon^{\text{up}}_\omega\right) \nonumber\\
     &\qquad + \int_{\partial S_-} t\cdot \pi\left( \left.\frac{\dd}{\dd\omega}\right|_{\omega=\omega_n} \Psi_2^{4/3} \mathcal{J} \Upsilon^{\text{in}}_\omega, \Upsilon^{\text{up}}_{\omega_n} \right).
  \end{align}
  By lemma~\ref{lemma:wronskian-boundary} we can write the second of
  these terms as the derivative of the Wronskian,
  \begin{align}
    &\left.\frac{\dd}{\dd\omega}\right|_{\omega=\omega_n} \int_{\partial S_-} t\cdot\pi\left(\Psi_2^{4/3} \mathcal J \Upsilon^{\text{in}}_\omega, \Upsilon^{\text{up}}_\omega\right) \nonumber\\
     &\quad= 8 \pi M^{4/3} \left.\frac{\dd}{\dd\omega} \mathcal W[R^{\text{in}}_\omega, R^{\text{up}}_\omega]\right|_{\omega = \omega_n}.
  \end{align}

  Summarizing our results so far, we have shown that
  \begin{align}\label{eq:balance2}
    &8 \pi M^{4/3} \left.\frac{\dd}{\dd\omega} \mathcal W[R^{\text{in}}_\omega, R^{\text{up}}_\omega]\right|_{\omega = \omega_n} \nonumber\\
    &\quad = -i \int_S \pi(\Psi_2^{4/3} \mathcal J \Upsilon^{\text{in}}_{\omega_n}, \Upsilon^{\text{up}}_{\omega_n}) \nonumber\\
    &\qquad - \int_{\partial S_-} t\cdot \pi\left( \left.\frac{\dd}{\dd\omega}\right|_{\omega=\omega_n} \Psi_2^{4/3} \mathcal{J} \Upsilon^{\text{in}}_\omega, \Upsilon^{\text{up}}_{\omega_n} \right)
    \nonumber\\
     &\qquad - \int_{\partial S_+} t\cdot\pi\left(\Psi_2^{4/3} \mathcal J \Upsilon^{\text{in}}_{\omega_n}, \left. \frac{\dd}{\dd\omega} \right|_{\omega=\omega_n} \Upsilon^{\text{up}}_\omega\right).
  \end{align}
As $S \to \Sigma_\mathbb{C}$, the boundary integrals vanish exponentially, so the right hand side
of~\eqref{eq:balance2} reduces to  $-i \llangle \Upsilon^{\text{in}}_{\omega_n}, \Upsilon^{\text{up}}_{\omega_n} \rrangle$. 
\end{proof}

The desired equivalence between  \eqref{eq:qnm exc} and \eqref{eq:excitation coeff} can be seen immediately by
substituting \eqref{eq:I source modes} for ${}_s I_{\ell m n}$, comparing with \eqref{eq:Kinnersley-bilinear}, and  applying lemma~\ref{lemma:Wronskian-derivative}.

{We checked numerically that the excitation factors obtained via the bilinear form, Eq.~\eqref{eq:excitation_factor}, are in good agreement (for some $a$, $n$ within a few \%, which we believe arises from inaccuracies of our mode functions close to the horizon) with the ones computed via direct differentiation of the Wronskian in Refs.~\cite{Zhang:2013ksa,Oshita2021}. }

\section{Concluding remarks}

We end this paper with some potential applications and alternatives to our formalism. 

The main motivation of this work is to provide some tools needed to study the black hole ringdown beyond linear order in perturbation theory. Higher orders are already needed to interpret high-precision numerical relativity simulations of binary mergers~\cite{Mitman:2022qdl,Cheung:2022rbm}, and could be needed to analyse gravitational wave observations by future detectors. 
Roughly speaking, we wish to make an ansatz 
for the solution of the non-linear system as a linear combination of quasinormal modes with \emph{time dependent} excitation coefficients similar to \eqref{excited}. The idea is that the bilinear form will help us writing down a dynamical system for these coefficients by analogy with wave equations on compact spaces, where the normal modes would be used instead to compute the overlap integrals required for terms in this dynamical system that are non-linear in the modes.

As extremality is approached, it is well known that a family of quasinormal modes becomes arbitrarily long-lived, with
$\Re \omega \approx m\Omega_H$
\cite{PressTeukolsky1973,Detweiler1977,Leaver1985,Hod:2008zz,Yang:2012pj,Cook:2014cta}. With
a commensurate frequency spectrum and arbitrarily slow decay, these
modes have been conjectured to become \emph{turbulent} as
$a\to M$~\cite{Yang:2014tla}. {By considering the nonlinear excitation coefficients near extremality using the approach detailed above, one may hope to establish (or rule out) the emergence of turbulent behavior, or to get a new perspective on the Aretakis instability \cite{aretakis2012decay}. It would also be interesting to see if our bilinear form can be used directly \emph{at} extremality. This would involve understanding better the behavior of the complex contour integrals that we employ to regulate the bilinear form at extremality.\footnote{{For recent work on mode solutions to the Teukolsky equation at extremality, see \cite{onozawa1996quasinormal,richartz2016quasinormal}.}}}

Applications along similar lines could include mode
mixing in clouds of ultralight scalar fields that could form outside
Kerr black holes~\cite{Arvanitaki:2010sy}. Here, once these clouds grow via
the superradiant instability, nonlinear interactions between the modes have been conjectured to give rise to a coherent
emission of gravitational waves or a bosenova~\cite{Yoshino:2013ofa}, see~\cite{Baumann:2018vus,Baumann:2022pkl} for recent proposals based on heuristic methods. It would be interesting to see whether our methods could be used to conceptualize or shed more light on the theoretical basis of such proposals. 

In \cite{gajic2021quasinormal}, a different approach is taken to quasinormal modes. Their essential idea is to consider, instead of a time $t$
Cauchy surface intersecting the bifurcation surface and spatial infinity, a ``hyperboloidal'' slice intersecting the future event horizon and future null infinity. On such a slice, they define a certain space of ``almost analytic'' functions (a ``Gevrey space'') encoding somehow the ``boundary conditions'' \eqref{eq:R bcs}. Their space is in fact a genuine Hilbert space, and the time evolution is represented on this space by a semigroup whose generator is essentially the Hamiltonian, $\mathcal{H}$ (see appendix \ref{sec:Lagrangian-Hamiltonian}). Their inner product is non-canonical -- and is not conserved -- and the generator of the semi-group is correspondingly not symmetric, as is also not physically expected due to the ``dissipative'' nature of quasinormal modes. Nevertheless, their analysis shows that quasinormal modes are genuine eigenfunctions $\mathcal{H} \Upsilon = i\omega \Upsilon$ in this space -- crucially, by contrast to their restriction to a constant $t$ surface, they do not blow up on the hyperboloidal slice as the horizon or scri are approached. 
While the quasinormal modes are not orthogonal in their inner product, the definition of our bilinear form with a hyperboloidal slice is also clearly possible and it would be interesting to see whether quasinormal modes, as defined in the 
setting of \cite{gajic2021quasinormal} (see also~\cite{Zenginoglu:2011jz,PanossoMacedo:2019npm,Ripley:2022ypi}) are still orthogonal in this bilinear form, as we conjecture. This would provide an alternative to our regularization procedure involving complex contours. 

Finally, it will be interesting to explore the relation between our bilinear form and the adjoint-spheroidal functions introduced in Ref.~\cite{London:2020uva}.

\medskip
{\bf Acknowledgements:} We thank M.~Casals, E.~Flanagan, D.~Gajic and A.~Grant for comments and discussions related to this work. We thank J.~Lestingi, S.~Ma and H.~Yang for pointing out a typo in an earlier version of the manuscript. SH thanks the Max-Planck Society for supporting the collaboration between MPI-MiS and Leipzig U., Grant Proj. Bez. M.FE.A.MATN0003. VT is grateful to the International Max Planck Research School, MPI-MiS for support through a studentship. This work makes use of the Black Hole Perturbation Toolkit.

\appendix

\begin{widetext}
\section{Proof of \eqref{eq:intertwine} and explicit form of $H^{ab}$}
\label{app:A}

For the current calculations, we use the following forms for the Einstein operator $\mathcal{E}$, the separation operator $\mathcal{S}$, $\mathcal{T}$ and the Teukolsky operator $\mathcal{O}$. 
In the expressions below, the GHP operators $\eth = m^a \Theta_a, \eth' = \bar m^a \Theta_a, \thorn = l^a \Theta_a, \thorn'=n^a \Theta_a$ are understood and their commutation relations as well 
as the algebraically special properties of an NP tetrad aligned with the principal null directions in a Petrov type D spacetime (implying that $\sigma = \sigma' = \kappa = \kappa' = 0, \Psi_i = 0, i \neq 2$ in GHP notation) are heavily used in all subsequent manipulations.  
\begin{subequations}
\begin{multline}
(\mathcal{E} h)_{ab} = \half (\nabla^c \nabla_a h_{bc} + \nabla^c \nabla_b h_{ac} - \nabla_a \nabla_b {h^c}_c - \nabla^c \nabla_c h_{ab} - g_{ab} \nabla^c \nabla^d h_{cd} + g_{ab} \nabla^c \nabla_c {h^d}_d)\\
= P^{cabdef} \nabla_c \nabla_d h_{ef},
\end{multline}
\begin{multline}
\mathcal{S} T = - (\eth - \bar{\tau'} - 4\tau)(\eth - \bar{\tau'}) T_{ll} - (\thorn - \bar{\rho} - 4\rho)(\thorn - \bar{\rho}) T_{mm} \\
+ \left[ (\eth - \bar{\tau'} - 4\tau)(\th - 2\bar{\rho})  + (\thorn - \bar{\rho} - 4\rho)(\eth - 2\bar{\tau'}) \right] T_{lm},
\end{multline}
\begin{multline}
\mathcal{T} h = \half (\eth - \bar{\tau'})(\eth - \bar{\tau'}) h_{ll} + \half (\thorn - \bar{\rho})(\thorn - \bar{\rho}) h_{mm} - \half \left[(\thorn - \bar{\rho})(\eth - 2 \bar{\tau'}) + (\eth -\bar{\tau'})(\thorn - 2 \bar{\rho})\right] h_{lm} \\
= - \psi_0,
\end{multline}
\begin{equation}
\mathcal{O} \eta = 2 \left[ (\thorn - 4\rho - \bar{\rho})(\thorn'-\rho') - (\eth - 4\tau - \bar{\tau'})(\eth' - \tau') -3\Psi_2 \right] \eta,
\end{equation}
\end{subequations}
where
\begin{equation}
P^{abcdef} = \half g^{ad} g^{bc} g^{ef} + \half g^{ae} g^{bd} g^{cf}  + \half g^{ae} g^{bf} g^{cd} - \half g^{ab} g^{cd} g^{ef} - \half g^{ad} g^{be} g^{cf} - \half g^{af} g^{bc} g^{de}.
\end{equation}
The corresponding adjoint operators satisfy
\begin{subequations}\label{FormaldefsofConjugates}
\begin{equation}
\nabla_a w^a[h, h'] = h_{bc}(\mathcal{E} h')^{bc} - h'_{bc}(\mathcal{E}^\dagger h)^{bc},
\end{equation}
\begin{equation}
\nabla_a \sigma^a[\phi, T] = \phi (\mathcal{S} T) - T_{a b} (\mathcal{S}^\dagger \phi)_{a b},
\end{equation}
\begin{equation}
\nabla_a t^a[\eta, h] = \eta (\mathcal{T} h) - h^{a b} (\mathcal{T}^\dagger \eta)_{a b},
\end{equation}
\begin{equation}
\nabla_a \pi^a[\phi, \eta] = \phi (\mathcal{O} \eta) - \eta (\mathcal{O}^\dagger \phi).
\end{equation}
\end{subequations}
which defines $w^a, \sigma^a, t^a, \pi^a$ up to a total divergence. Below, 
we make specific choices, which will lead to a specific $H^{ab}$.
One can derive the adjoint operators explicitly in GHP form as
\begin{subequations}
\begin{equation}
(\mathcal{E}^\dagger h)_{ab} = (\mathcal{E} h)_{ab},
\end{equation}
\begin{multline}
(\mathcal{S}^\dagger \phi)_{a b} = - l_a l_b(\eth - \tau)(\eth + 3\tau) \phi - m_a m_b (\thorn - \rho)(\thorn + 3\rho) \phi \\
+ l_{(a} m_{b)} \{(\thorn - \rho + \bar{\rho})(\eth + 3\tau) + (\eth - \tau + \bar{\tau'})(\thorn + 3\rho) \} \phi \\
= \nabla^c \big( (l_{(a} m_{b)} l_c - l_a l_b m_c) (\eth + 3\tau)\phi + (l_{(a} m_{b)} m_c - m_a m_b l_c) (\thorn + 3\rho) \phi \big) \\
= \nabla^c (\nabla^d - 4 \tau \overline{m}^d + 4 \rho n^d) \big( \phi (2 l_{(a} m_{b)} l_{(c} m_{d)} - l_a l_b m_c m_d - m_a m_b l_c l_d) \big),
\end{multline}
\begin{multline}
(\mathcal{T}^\dagger \eta)_{a b} = \half l_a l_b (\eth - \tau)(\eth - \tau) \eta + \half m_a m_b (\thorn - \rho)(\thorn - \rho) \eta \\
- \half l_{(a} m_{b)}\{(\eth + \bar{\tau'} -\tau)(\thorn - \rho) + (\thorn - \rho + \bar{\rho})(\eth - \tau) \} \eta  \\
= -\half \nabla^c \nabla^d \big( \eta (2 l_{(a} m_{b)} l_{(c} m_{d)} - l_a l_b m_c m_d - m_a m_b l_c l_d) \big),
\end{multline}
\begin{equation}
\mathcal{O}^\dagger \phi = 2\left[ (\thorn' - \bar{\rho'})(\thorn + 3 \rho) - (\eth' - \bar{\tau})(\eth + 3\tau) -3\Psi_2 \right] \phi.
\end{equation}
\end{subequations}
When computing these, the boundary terms $w^a, \sigma^a, t^a, \pi^a$ drop out as
\begin{subequations}
\begin{equation}
w^a[h, h'] = h_{bc} \left(\mathcal{F} h' \right)^{abc} - h'_{bc} \left(\mathcal{F} h \right)^{abc},
\end{equation}
\begin{multline}
\sigma^a[\phi, T] = l^a \left[ T_{mm} (\thorn + 3 \rho) \phi - \phi (\thorn - \bar{\rho}) T_{mm} + \phi (\eth - 2 \bar{\tau'}) T_{lm} - T_{lm} (\eth + 3 \tau) \phi \right] \\
+ m^a \left[ T_{ll} (\eth + 3 \tau) \Phi - \Phi (\eth - \bar{\tau'}) T_{ll} + \Phi (\thorn - 2 \bar{\rho}) T_{lm} - T_{lm} (\thorn + 3 \rho) \Phi \right],
\end{multline}
\begin{multline}
t^a[\eta, h] = \half l^a \left[ \eta (\thorn - \bar{\rho}) h_{mm} - h_{mm} (\thorn - \rho) \eta + h_{lm} (\eth - \tau) \eta - \eta (\eth - 2 \bar{\tau'}) h_{lm} \right] \\
+\half m^a \left[ \eta (\eth - \bar{\tau'}) h_{ll} - h_{ll} (\eth - \tau) \eta + h_{lm} (\thorn - \rho) \eta - \eta (\thorn - 2 \bar{\rho}) h_{lm} \right],
\end{multline}
\begin{equation}
\pi^a[\phi, \eta] = \phi \left( \Theta^a + 4 B^a \right) \eta - \eta \left( \Theta^a - 4 B^a \right) \phi,
\end{equation}
\end{subequations}
where
\begin{equation}
\label{Fdefinition}
\left(\mathcal{F} h \right)^{abc} = P^{abcdef} \nabla_d h_{ef}.
\end{equation}
The operator $\mathcal{F}$ is related to the linearized Einstein operator by $(\mathcal{E} h)^{ab} = \nabla_c \left(\mathcal{F} h \right)^{cab}$.\par
Using the equation $\mathcal{O} \mathcal{T} = \mathcal{S} \mathcal{E}$, one obtains the following relation \eqref{wsigmapitidentity} between $w^a$,  $\sigma^a$, $\pi^a$, and $t^a$.
\begin{equation}
\label{wsigmapitidentity}
\nabla_a \left( w^a[\mathcal{S}^\dagger \Phi, h] + \sigma^a[\Phi, \mathcal{E} h] - \pi^a[\Phi, \mathcal{T} h] - t^a[\mathcal{O}^\dagger \Phi, h] \right) = 0.
\end{equation}
\cite{wald1990identically} have proven that one learns from this equation that there must exist a 2-form $H$, constructed out of the fields $h_{ab}, \Phi$ and their derivatives such that $w + \sigma - \pi - t = \intd \star H$, where $w = w_a(\mathcal{S}^\dagger \Phi, A) \star {\rm d}x^a$, $\sigma = \sigma_a(\Phi, \mathcal{E} A) \star {\rm d}x^a$, $\pi = \pi_a(\Phi, \mathcal{T} A) \star {\rm d}x^a$, and $t = t_a(\mathcal{O}^\dagger \Phi, A) \star {\rm d}x^a$.
In divergence form, the identity for $H^{ab}[\Phi, h]$ is
\begin{equation}\label{HabMastereq}
\nabla_b H^{ba}[\Phi, h] = - w^a[\mathcal{S}^\dagger \Phi, h] - \sigma^a[\Phi, \mathcal{E} h] + t^a[\mathcal{O}^\dagger \Phi, h] + \pi^a[\Phi, \mathcal{T} h].
\end{equation}
Of course, $H^{ab}$ is defined by the above equation only up to a total local divergence $\nabla_c C^{[abc]}$, or equivalently, $\star H$ is only defined up to an exact local 1-form.
Therefore, the expression for $H^{ab}$ given below only represents one possible solution. Unfortunatlely, the proof given in \cite{wald1990identically} is not really useful to actually find an $H^{ab}$ in practice, so we simply try to ``peel off'' the total divergence from the right side of \eqref{HabMastereq}
by hand. 

For this, it is convenient to write out the divergence operator on $H^{ab}$ in terms of GHP operators as
\begin{equation}\label{Divegenceof2form}
\begin{split}
\nabla_b H^{ba} =& \left[ n^a \left( \thorn - \rho - \bar{\rho} \right) - l^a \left( \thorn' - \rho' - \bar{\rho}' \right) + \bar{m}^a \left( \tau - \bar{\tau}' \right) - m^a \left( \tau' - \bar{\tau} \right) \right] H_{nl} \\
+& \left[ m^a \left( \thorn - \rho \right) - l^a \left( \eth - \tau \right) \right] H_{\bar{m} n} + \left[ \bar{m}^a \left( \thorn - \bar{\rho} \right) - l^a \left( \eth' - \bar{\tau} \right) \right] H_{m n} \\
+& \left[ m^a \left( \thorn' - \bar{\rho}' \right) - n^a \left( \eth - \bar{\tau}' \right) \right] H_{\bar{m} l} + \left[ \bar{m}^a \left( \thorn' - \rho' \right) - n^a \left( \eth' - \tau' \right) \right] H_{m l} \\
+& \left[ \bar{m}^a \left( \eth - \tau - \bar{\tau}' \right) - m^a \left( \eth' - \tau' - \bar{\tau} \right) + n^a \left( \rho - \bar{\rho} \right) - l^a \left( \rho' - \bar{\rho}' \right) \right] H_{\bar{m} m}.
\end{split}
\end{equation}
We now substitute the above expressions for $\mathcal{S}^\dagger \Phi$, $\mathcal{E} h$, $\mathcal{O}^\dagger \Phi$ and $\mathcal{T} h$ respectively into $w^a$, $\sigma^a$, $t^a$ and $\pi^a$ on the right side of \eqref{HabMastereq}. Next, by comparing the result with \eqref{Divegenceof2form} we should in principle be able to get the components of the 2-form $H^{ab}[\Phi, h]$ up to a total divergence. 
This calculation is extremely tedious and was done using the Xtensor package of Mathematica, along the following lines. 
First, we neglect the terms containing less than three derivatives on the right side of \eqref{HabMastereq} and try to ``peel off'' a total divergence by making suitable choices of the
NP components of $H^{ab}$ appearing in \eqref{Divegenceof2form}. In fact, we cannot peel off a divergence exactly, but only at the cost of terms containing less than three derivatives. Having determined 
the highest derivative parts of the NP components of $H^{ab}$, we collect any terms with less than three derivatives that were left over when peeling of the total divergence, and combine them 
with all terms with less than three derivatives  on the right side of \eqref{HabMastereq} that were neglected so far. Of these remaining terms, we now discard all terms with less than two derivatives and repeat the process until no terms are left. We are ensured on general grounds that this must happen, though it is in practice rather non-trivial and time-consuming to determine all terms. {The corresponding Mathematica notebooks are provided at \cite{Vahid_github}}.  All in all, it is found that a skew symmetric tensor $H^{ab}$ which satisfies \eqref{HabMastereq} is given by the following expression.
\begin{equation}\label{ExpliciteHab}
\begin{split}
H^{ab} =& \Big\{ - \half h_{ll} (\mathcal{S}^\dagger \Phi)_{nn} + h_{lm} (\mathcal{S}^\dagger \Phi)_{n\overline{m}} - \half h_{mm} (\mathcal{S}^\dagger \Phi)_{\overline{m}\overline{m}} - \rho h_{lm} (\eth + 3\tau) \Phi + \mathcal{T} h \Phi\\
&- (\left(\mathcal{F} h \right)_{mlm} - \left(\mathcal{F} h \right)_{lmm} - \tau h_{lm}) (\thorn + 3\rho) \Phi \Big\} \left( l^a n^b - n^a l^b \right) \\
+& \Big\{ - h_{ln} (\mathcal{S}^\dagger \Phi)_{n\overline{m}} + h_{l\overline{m}} (\mathcal{S}^\dagger \Phi)_{nn} - h_{m\overline{m}} (\mathcal{S}^\dagger \Phi)_{n\overline{m}} + h_{nm} (\mathcal{S}^\dagger \Phi)_{\overline{m}\overline{m}} + (\mathcal{I} h) (\eth + 3 \tau) \Phi \\
&+ (\mathcal{J} h) (\thorn + 3 \rho) \Phi + (\mathcal{K} h) \Phi \Big\} \left( l^a m^b - m^a l^b \right) \\
+& \Big\{ (\left(\mathcal{F} h \right)_{mlm} - \left(\mathcal{F} h \right)_{lmm} + \rho h_{mm}) (\eth + 3 \tau) \Phi - \tau h_{mm} (\thorn + 3 \rho) \Phi \Big\} \left( l^a \overline{m}^b - \overline{m}^a l^b \right) \\
+& \Big\{ (\left(\mathcal{F} h \right)_{llm} - \left(\mathcal{F} h \right)_{mll} + \tau h_{ll}) (\thorn + 3 \rho) \Phi - \rho h_{ll} (\eth + 3 \tau) \Phi \Big\} \left( n^a m^b - m^a n^b \right) \\
+& \Big\{ \half h_{ll} (\mathcal{S}^\dagger \Phi)_{nn} - h_{lm} (\mathcal{S}^\dagger \Phi)_{n\overline{m}} + \half h_{mm} (\mathcal{S}^\dagger \Phi)_{\overline{m}\overline{m}} + \tau h_{lm} (\thorn + 3 \rho) \Phi - \mathcal{T} h\\
&+ (\left(\mathcal{F} h \right)_{llm} - \left(\mathcal{F} h \right)_{mll} - \rho h_{lm}) (\eth + 3 \tau) \Phi\Phi \Big\} \left( m^a \overline{m}^b - \overline{m}^a m^b \right),
\end{split}
\end{equation}
where
\begin{equation}
\begin{split}
\mathcal{I} h =& \left( \thorn + \rho - \bar{\rho} \right) h_{m\bar{m}} - \left( \eth - \tau' \right) h_{l\bar{m}} - \rho' h_{ll} - (\rho + \bar{\rho}) h_{ln} + \tau' h_{lm},\\
\mathcal{J} h =& \left( \eth + \tau - \tau' \right) h_{ln} -\left( \thorn - \bar{\rho} \right) h_{nm} + \rho' h_{lm} - \tau' h_{mm} - (\tau + \bar{\tau}') h_{m\bar{m}},\\
\mathcal{K} h =& \half \left[ \eth \thorn' - \bar{\tau}' \thorn' + \left( 3 \rho' - \bar{\rho}' \right) \eth + 2 \left( \rho' + \bar{\rho}' \right) \tau \right] h_{ll}\\
-& \half \left[ \eth \thorn + \left( \tau - 2 \bar{\tau}' \right) \thorn + \left( \rho - \bar{\rho} \right) \eth + 2 \left( \rho \tau - \bar{\rho} \bar{\tau}' \right) \right] h_{ln}\\
-& \half \bigg[ \thorn' \thorn + \eth' \eth + \left( 2 \rho' - \bar{\rho}' \right) \thorn + \left( \rho - 2 \bar{\rho} \right) \thorn' + \left( 2 \tau' - \bar{\tau} \right) \eth + \left( \tau - 2 \bar{\tau}' \right) \eth'\\
&+ 2 \left( 2 \tau \tau' - \bar{\tau} \bar{\tau}' - 2 \tau \bar{\tau} - \bar{\rho} \bar{\rho}' - 2 \rho' \bar{\rho} + 2 \rho \rho' - \half \bar{\Psi}_2 \frac{\rho}{\bar{\rho}} + \half \Psi_2 \frac{\bar{\rho}}{\rho} \right) \bigg] h_{lm}\\
+& \half \left[ \eth \eth + 2 \left( \tau - \bar{\tau}' \right) \eth + 2 \tau \left( \tau - \bar{\tau}' \right) \right] h_{l\bar{m}}\\
+& \half \left[ \thorn \thorn + 2 \left( \rho - \bar{\rho} \right) \thorn + 2 \rho \left( \rho - \bar{\rho} \right) \right] h_{nm}\\
+& \half \left[ \eth' \thorn - \left( \bar{\tau} - 2 \tau' \right) \thorn + \left( \rho - \bar{\rho} \right) \eth' + 2 \rho \left( \bar{\tau} + 2 \tau' \right) \right] h_{mm}\\
-& \half \left[ \eth \thorn + \left( \tau - 2 \bar{\tau}' \right) \thorn + \left( \rho - \bar{\rho} \right) \eth + 2 \left( \rho \tau - \bar{\rho} \bar{\tau}' \right) \right] h_{m\bar{m}}\\
%\overset{?}{=}& (\mathcal{E} h)_{lm} - 2 \psi_1.
\end{split}
\end{equation}
When integrating $H^{ab}$ over a 2-surface we need the Hodge-dual, which is given by
\begin{equation}
\begin{split}
\star H =& -i \Big\{ \half h_{ll} (\mathcal{S}^\dagger \Phi)_{nn} - h_{lm} (\mathcal{S}^\dagger \Phi)_{n\overline{m}} + \half h_{mm} (\mathcal{S}^\dagger \Phi)_{\overline{m}\overline{m}} + \tau h_{lm} (\thorn + 3 \rho) \Phi - \mathcal{T} h \Phi\\
&+ (\left(\mathcal{F} h \right)_{llm} - \left(\mathcal{F} h \right)_{mll} - \rho h_{lm}) (\eth + 3 \tau) \Phi\Big\} \left( l \wedge n \right)\\
+& i \Big\{ - h_{ln} (\mathcal{S}^\dagger \Phi)_{n\overline{m}} + h_{l\overline{m}} (\mathcal{S}^\dagger \Phi)_{nn} - h_{m\overline{m}} (\mathcal{S}^\dagger \Phi)_{n\overline{m}} + h_{nm} (\mathcal{S}^\dagger \Phi)_{\overline{m}\overline{m}} + (\mathcal{I} h) (\eth + 3 \tau) \Phi \\
&+ (\mathcal{J} h) (\thorn + 3 \rho) \Phi + (\mathcal{K} h) \Phi \Big\} \left( l \wedge m \right) \\
-& i \Big\{ (\left(\mathcal{F} h \right)_{mlm} - \left(\mathcal{F} h \right)_{lmm} + \rho h_{mm}) (\eth + 3 \tau) \Phi - \tau h_{mm} (\thorn + 3 \rho) \Phi \Big\} \left( l \wedge \overline{m} \right) \\
-& i \Big\{ (\left(\mathcal{F} h \right)_{llm} - \left(\mathcal{F} h \right)_{mll} + \tau h_{ll}) (\thorn + 3 \rho) \Phi - \rho h_{ll} (\eth + 3 \tau) \Phi \Big\} \left( n \wedge m \right) \\
-& i \Big\{ - \half h_{ll} (\mathcal{S}^\dagger \Phi)_{nn} + h_{lm} (\mathcal{S}^\dagger \Phi)_{n\overline{m}} - \half h_{mm} (\mathcal{S}^\dagger \Phi)_{\overline{m}\overline{m}} - \rho h_{lm} (\eth + 3\tau) \Phi + \mathcal{T} h \Phi\\
&- (\left(\mathcal{F} h \right)_{mlm} - \left(\mathcal{F} h \right)_{lmm} - \tau h_{lm}) (\thorn + 3\rho) \Phi\Big\} \left( m \wedge \overline{m} \right).
\end{split}
\end{equation}
We remind the reader that $\star H$ is unique only up to an exact local form. 
Such a contribution could be introduced e.g. to simplify the form of $\star H$. We have not attempted to do this.
Finally, we remark that for spin-1, the tensor $H^{ab}$ has already been constructed in \cite{hollands2020radiation}. It seems to have a much simpler form than in the spin-2 case.

\section{Proof of \eqref{commutator}}
\label{app:B}

In this appendix, we will prove \eqref{commutator}. Although that relation is used in the body of the paper only for gravitational perturbations, i.e., spin $s=(p-q)/2=\pm 2$, 
most of the calculations work for arbitrary $s$, and in fact $p$ and $q$. We therefore state the results in their most general form in view of 
their potential applicability to various values of the spin and for general Petrov type D spacetimes, of which Kerr is an example.

We begin by defining
\begin{equation}\label{eq:Ohatoperatordef}
\mathcal{P} \eta = \left( \Theta_a + p B'_a + q \bar{B}'_a \right) \left( \Theta^a + p B'^a + q \bar{B}'^a \right) \eta + 2 (\dfrac{p}{\zeta} + \dfrac{q}{\bar{\zeta}} ) \GHPLie_\xi \eta,
\end{equation}
where $B_a$ was given above in \eqref{Bdef}, $\zeta$ was given in \eqref{zetadef}, $\xi^a$ was given in \eqref{xidef}, 
and where here and in the following we use the GHP priming operation. Strictly speaking, $\mathcal P$
denotes not one operator, but sevaral operators depending on the chosen $(p,q)$, and the same is understood for other operators below. 

The following lemma is checked by direct computation.

\begin{lemma}\label{lemma:OhatandKproperties}
The operators $\mathcal{P}$ and $\mathcal{K}$ (see \eqref{eq:Koperatordef}) have the following properties:
\begin{enumerate}[label=(\roman*)]
\item They are self-dual, i.e., $\mathcal{K}^\dagger = \mathcal{K}$ and $\mathcal{P}^\dagger = \mathcal{P}$,
\item They are real, i.e.,  $\overline{\mathcal{K} \eta} = \mathcal{K} \bar{\eta}$ and $\overline{\mathcal{P} \eta} = \mathcal{P} \bar{\eta}$ for every smooth weighted scalar $\eta$, where we mean the GHP overbar operation,
\item $\mathcal{K}' \eta' = \zeta^{-p} \bar{\zeta}^{-q} \mathcal{K} \zeta^p \bar{\zeta}^q \eta'$ and $\mathcal{P}' \eta' = \zeta^{-p} \bar{\zeta}^{-q} \mathcal{P} \zeta^p \bar{\zeta}^q \eta'$ for any $\eta  \GHPwt (p,q)$,
\item $\mathcal{P} \Phi = \mathcal{O}^\dagger \Phi$ for every smooth weighted scalar $\Phi  \GHPwt (-2s,0)$,
\item $\mathcal{P} \phi = \mathcal{O} \phi$ for every smooth weighted scalar $\phi  \GHPwt (2s,0)$
\end{enumerate}
where $s=(p-q)/2$ is the spin of the theory, i.e., $s=1$ for electromagnetic- and $s=2$ for gravitational perturbations.
\end{lemma}
Next, we construct the following operators $\mathcal{G}$, $\mathcal{R}$ and $\mathcal{L}$ by combining $\mathcal{K}$ and $\mathcal{P}$ in a certain way.

\begin{equation}
\begin{split}
\mathcal{G} =& \mathcal{K} - \dfrac{1}{4} K \mathcal{P},\\
\mathcal{R} =& \dfrac{1}{4} \abs{\zeta}^2 \mathcal{P} + \mathcal{G},\\
\mathcal{L} =& \dfrac{1}{4} \abs{\zeta}^2 \mathcal{P} - \mathcal{G}.
\end{split}
\end{equation}
Here $K = g^{ab}K_{ab}$ is the trace of the Killing tensor $K_{ab}$ \eqref{Kabdef}. See \eqref{eq:Koperatordef} for the definition of the ``Carter'' operator $\mathcal K$ for general $(p,q)$.
We will make use of the following lemma \ref{lemma:OKcommutator} observed by \cite{beyer2008new}.
\begin{lemma}\label{lemma:OKcommutator}
Suppose we have partial differential operators $\mathcal{A}, \mathcal{B}$ (defined on suitably compatible vector bundles) such that
\begin{enumerate}[label=(\roman*)]
\item $[\mathcal{A}, \mathcal{B}] = 0$,
\item $\mathcal{A}(\alpha f) = \alpha \mathcal{A}(f)$,
\item $\mathcal{B}(\beta f) = \beta \mathcal{B}(f)$
\end{enumerate}
for all $f$, where $\alpha,\beta$ are certain functions, i.e., multiplication operators. Then
\begin{equation}
\left[ \frac{1}{\alpha + \beta} (\mathcal{A} + \mathcal{B}), \frac{\alpha}{\alpha + \beta} \mathcal{A} - \frac{\beta}{\alpha + \beta} \mathcal{B} \right] = 0.
\end{equation}
\end{lemma}
This lemma is used to prove the following result.  

\begin{theorem}\label{lemma:DLcommutator}
Let $\eta \circeq \GHPw{p}{q}$. Then in any type D spacetime
\begin{enumerate}[label=(\roman*)]
\item $\left[ \mathcal{K}, \mathcal{P} \right] = 0$,
\item $\abs{\zeta}^{-2} \mathcal{G} \abs{\zeta}^2 = \mathcal{G}^\dagger$,
\item $\left[ \abs{\zeta}^2 \mathcal{P}, \mathcal{G} \right] = 0$.
\end{enumerate}
\end{theorem}
\begin{proof}\label{proof:DLcommutator}
Let $\alpha := K_{ln} = - \dfrac{1}{8} \left( \zeta - \bar{\zeta} \right)^2$, $\beta := K_{m\bar{m}} = \dfrac{1}{8} \left( \zeta + \bar{\zeta} \right)^2$ and
\begin{equation}
\begin{split}
\mathcal{A} := \mathcal{R} \equiv& \half \abs{\zeta}^2 \left( \thorn - \rho - \bar{\rho} \right) \left( \thorn' - p \rho' - q \bar{\rho}' \right) + \half \abs{\zeta}^2 \left( \thorn' - (p+1) \rho' - (q+1) \bar{\rho}' \right) \thorn\\
&+ \half (p+q) (\zeta + \bar{\zeta}) \GHPLie_{\xi},\\
\mathcal{B} := \mathcal{L} \equiv& - \half \abs{\zeta}^2 \left( \eth - (q+1) \bar{\tau}' - \tau  \right) \left( \eth' - p \tau' \right) - \half \abs{\zeta}^2 \left( \eth' - (p+1) \tau' - \bar{\tau} \right) \left( \eth - q \bar{\tau}' \right)\\
&- \half (p-q) (\zeta - \bar{\zeta}) \GHPLie_{\xi}.
\end{split}
\end{equation}
Through a long tedious calculation using heavily the type D property of the background, one can show that
\begin{enumerate}[label=(\roman*)]
\item $\left[ \mathcal{R}, \mathcal{L} \right] \eta = 0$,
\item $\mathcal{R} (K_{ln} \eta) = K_{ln} \mathcal{R} \eta$,
\item $\mathcal{L} (K_{m\bar{m}} \eta) = K_{m\bar{m}} \mathcal{L} \eta$.
\end{enumerate}
Consequently, using the lemma \ref{lemma:OKcommutator}, we conclude that
\begin{equation}
\left[ 2 \abs{\zeta}^{-2} (\mathcal{R} + \mathcal{L}), 2 \abs{\zeta}^{-2} K_{ln} \mathcal{R} - 2 \abs{\zeta}^{-2} K_{m\bar{m}} \mathcal{L} \right] = \left[ \mathcal{P}, \mathcal{K} \right] = 0.
\end{equation}
Further,
\begin{multline}
\mathcal{G}^\dagger = \mathcal{K} - \dfrac{1}{4} \mathcal{P} K = \frac{\alpha}{\alpha + \beta} \mathcal{R} - \frac{\beta}{\alpha + \beta} \mathcal{L} - \frac{1}{\alpha + \beta} (\mathcal{R} + \mathcal{L}) \frac{\alpha - \beta}{2}\\
= \frac{1}{\alpha + \beta} \mathcal{R} \alpha - \frac{1}{\alpha + \beta} \mathcal{L} \beta - \frac{1}{\alpha + \beta} (\mathcal{R} + \mathcal{L}) \frac{\alpha - \beta}{2} = \frac{1}{\alpha + \beta} (\mathcal{R} - \mathcal{L}) \frac{\alpha + \beta}{2}\\
= \abs{\zeta}^{-2} \mathcal{G} \abs{\zeta}^2.
\end{multline}
In addition, we can also conclude that
\begin{equation}
\left[ \mathcal{R} + \mathcal{L}, \mathcal{R} - \mathcal{L} \right] = \left[ \half \abs{\zeta}^2 \mathcal{P}, 2 \mathcal{G} \right] = \left[ \abs{\zeta}^2 \mathcal{P}, \mathcal{G} \right] = 0.
\end{equation}
This demonstrates the claimed relations in the theorem.
\end{proof}

\section{Trivial conservation laws}\label{sec:trivial}

An important question when deriving conservation laws is whether the current is \emph{nontrivially} conserved. A conservation law is said to be trivial if either (a) the conserved quantity itself vanishes on solutions, or (b) the quantity is conserved even on non-solutions~\cite{olver1993applications}. A generic trivial conserved quantity will be a combination of (a) and (b), e.g., if by adding a term involving the equation of motion to a conserved current, it becomes conserved even on non-solutions, then the conservation law is trivial. 

As an example, consider a real Klein-Gordon field $\psi$, and a second order symmetry operator $\Lie_\xi^2$ for a Killing vector $\xi^a$. The associated current $x^a(\psi, \Lie_\xi^2\psi)$ is conserved as follows:
\begin{eqnarray}
\nabla^a x_a(\psi, \Lie_\xi^2\psi) &=& (\mathcal{X} \psi) (\Lie_\xi^2 \psi) - \psi (\mathcal{X} \Lie_\xi^2 \psi) \nonumber\\
&=& (\mathcal{X} \psi) (\Lie_\xi^2 \psi) - \psi (\Lie_\xi^2 \mathcal{X}  \psi) \nonumber\\
&=& \Lie_\xi\left[(\mathcal{X} \psi) (\Lie_\xi \psi) - (\Lie_\xi \mathcal{X} \psi) \psi\right] \nonumber\\
&=& \nabla^a \xi_a \left[ (\mathcal{X} \psi) (\Lie_\xi \psi) -  (\Lie_\xi \mathcal{X} \psi) \psi\right].
\end{eqnarray}
Thus if we add the quantity $\xi^a \left[ (\mathcal{X} \psi) (\Lie_\xi \psi) -  (\Lie_\xi \mathcal{X} \psi) \psi\right]$---which vanishes on solutions $\mathcal X \psi = 0$---to $x^a(\psi, \Lie_\xi^2\psi)$, then the resulting current is identically conserved. Hence the current associated to $\Lie_\xi^2$ is trivial. It is easy to see that in this context symmetry operators $\Lie_\xi^n$ for odd $n$ (i.e., skew-adjoint operators) generate nontrivial conserved quantities, whereas those with even $n$ generate trivial ones.

In general, for a symmetry to generate a nontrivial conservation law, it must be a symmetry not just of the equation of motion, but also the Lagrangian. A complete analysis of triviality for the conserved currents discussed in the present manuscript is left for future work.

\section{WKB analysis of conserved currents}
\label{app:WKB}

We make a WKB (high frequency) ansatz for the metric perturbation of the usual form
\begin{equation}
\label{WKB}
h_{ab} \sim {\rm Re} \, A_{ab} e^{i\omega S}    
\end{equation}
with real frequency $\omega \gg 1$. For $h_{ab}$ to be a solution to the linearized Einstein equation, $S$ should as usual satisfy the eikonal equation $g^{ab }\nabla_a S \nabla_b S = 0$, so that $k^a = \nabla^a S$ is tangent to a future directed -- by assumption -- congruence of affine null geodesics. We may take the (complex) amplitude $A_{ab}$ to be varying slowly on time scales of order one, and supported in a very small tube around such a geodesic.  Imposing the usual transverse-traceless gauge conditions, $A_{ab}$ is found to obey a set of transport equations which are solved order by order in $\omega^{-1}$, see e.g. \cite{green2016superradiant}. In addition, $k^a A_{ab} = g^{ab} A_{ab} = 0$, i.e. $A_{ab}$ is a \emph{polarization tensor}. To be precise, when the expansion for $A_{ab}$ is carried out up to order $\omega^{-N}$, we only get a solution to the linearized Einstein equation up that order, which must be supplemented with a non-perturbative (in $\omega^{-1}$) correction of order $O(\omega^{-N})$. This correction can be chosen such that $h_{ab}$ is still supported near the world-tube of the geodesic in the vicinity of a given Cauchy surface $\Sigma$, see \cite{green2016superradiant} for details of this  non-trivial construction.

When a derivative $\nabla_a$ hits $A_{ab} e^{i\omega S}$, the oscillations dominate, corresponding effectively to the substitution
\begin{equation}
  \nabla_a \to ip_a \equiv i\omega k_a  
\end{equation}
familiar from quantum mechanics.
Under this substitution (also called the ``principal symbol'' in the mathematics literature), the operators $\mathcal{C}_n$ in \eqref{Cndef} become
\begin{equation}
\mathcal{C}_{n \, ab}{}^{cd} \to  -\tfrac{1}{2} \zeta^4 [-Q(p)]^n e_{ab}(p) e^{cd}(p)', 
\end{equation}
where $Q(p) = K^{ab} p_a p_b$ is Carter's constant
which is conserved along the null geodesic on account of 
$\nabla_{(a} K_{bc)} = 0$, and where $e_{ab}(p)$ and $e_{ab}(p)'$ are the polarization tensors 
\begin{equation}
e_{ab}(p)=Z_{ac} Z_{bd}p^c p^d, \quad 
e_{ab}(p)'=Z_{ac}' Z_{bd}' p^c p^d
\end{equation} in which the polarization tensor properties follow from the definitions $Z = l \wedge m$ and $Z' = n \wedge \bar m$. 
Thus, in the high frequency limit, $\mathcal{C}_n$ is basically the $n$-th power of Carter's constant, dressed by the tensor structure $e_{ab}(p) e^{ cd}(p)'$ involving polarization tensors related to the principal null direction $l^a$ respectively $n^a$. For the GHP primed symmetry operator $\mathcal{C}_n'$, the roles of the polarization tensors should be reversed. 

We now wish to evaluate $\chi_{(n)}[h]$ \eqref{chinint} on our WKB solution $h_{ab}$ \eqref{WKB}. Using the definition of $Z^{ab}$ and of $Z^{ab \prime}$, we can see that
\begin{equation}
\begin{split}
\bar e^{ab}(p) e_{ab}(p) &= (l^a p_a)^4, \ 
\quad e^{ab}(p) e_{ab}(p)  = 0,\\
\bar e^{ab}(p)' e_{ab}(p)' &= (n^a p_a)^4, 
\quad e^{ab}(p)' e_{ab}(p)' = 0,
\end{split}
\end{equation}
in addition to the usual properites $e_{ab}(p) p^a = g^{ab} e_{ab}(p) = 0$, and similarly for the primed polarization tensor $e_{ab}'(p)$. This means that \begin{equation}
\epsilon_{ab}^+(p) = \tfrac{1}{\sqrt{2}} {\rm Re} \, e_{ab}(p)'/(n^c p_c)^2 \qquad \epsilon_{ab}^\times(p) = \tfrac{1}{\sqrt{2}} {\rm Im} \, e_{ab}(p)'/(n^c p_c)^2
\end{equation} 
form an orthonormal basis of (real) polarization tensors unless $p_a \propto n_a$ which we assume for simplicity is not the case. In particular, we can write the complex amplitude of our WKB approximation as
\begin{equation}
A_{ab} = A_+ \epsilon_{ab}^+ + A_\times \epsilon_{ab}^\times 
\end{equation}
up to a ``gauge-transformation''\footnote{A tensor of the form $\xi_{(a} p_{b)}$ where $p_a \xi^a = 0$.}, which does not matter since $j^a$ is gauge invariant.
Making use of such relations in \eqref{jdef} and \eqref{wadef}, we find that in the high frequency limit, the conserved quantity $\chi$ -- the flux of Carter current through $\Sigma$ --
on a WKB solution \eqref{WKB} is given to leading order in $\omega$ by 
\begin{equation}
\begin{split}
\chi_{(n)}[h] = & \int_\Sigma j^a_{(n)} \dd \Sigma_a \\
\sim &  \, \, -i(-1)^n \int_\Sigma p^a |\zeta|^8 
{\rm Im}(A_+ \bar A_\times)  (p \cdot l)^4 (p \cdot n)^4  Q(p)^n \, \dd \Sigma_a
\end{split}
\end{equation}
where $K^{ab} p_a p_b = Q(p)$ denotes the  Carter constant.
The Killing tensor \eqref{Kabdef} can be written as $K^{ab}=|\zeta|^2 l^{(a} n^{b)} - \tfrac12 ({\rm Re} \zeta)^2 g^{ab}$ and since 
$g^{ab}p_a p_b=0$, we find that $|\zeta|^8 (p \cdot l)^4 (p \cdot n)^4 = Q(p)^4$. This leads to \eqref{chinint}.

The integrand is by construction sharply localized at the point where the geodesic pierces the Cauchy surface $\Sigma$, so the integral is basically the integrand at that point. 
Thus, $\chi_{(n)}[h]$ is essentially a power of the Carter constant for a WKB solution \cite{green2016superradiant} in the high frequency limit. An analogous result \cite{green2016superradiant}, holds for the canonical energy $E[h]$ \eqref{canen}: In that case we get the energy of a particle in the WKB limit.

Due to the presence of 
${\rm Im}(A_+ \bar A_\times)$ coupling the different polarizations, the Carter currents $j^a_{(n)}$ are of ``zilch''-type in the terminology of \cite{grant2020class}.

\section{Lagrangian and Hamiltonian for the Teukolsky
  equation}\label{sec:Lagrangian-Hamiltonian}

One can write down a 
Lagrangian for the Teukolsky equation:
\begin{equation}\label{eq:TeukolskyL}
  L( \tilde \Upsilon, \Upsilon ) = \epsilon \left[ g^{ab} (\Theta_a + 4 B_a) \tilde \Upsilon (\Theta_b - 4 B_b ) \Upsilon + 16 \Psi_2 \tilde \Upsilon \Upsilon \right],
\end{equation}
involving the fields
$\tilde \Upsilon\GHPwt (4, 0)$ and
$\Upsilon \GHPwt (-4,
0)$~\cite{Toth:2018ybm}. These fields should be varied independently
to obtain the field equations $\mathcal O^\dagger \Upsilon = 0$ and
$\mathcal O \tilde \Upsilon = 0$, respectively. The Lagrangian
obviously bears close resemblance to that of a $U(1)$ charged scalar
field; however, $B_a$ is not pure imaginary, and $\Upsilon$ and
$\tilde \Upsilon$ are not complex conjugates of each other.

Next we take the Legendre transform to obtain a Hamiltonian. We
foliate our Kerr spacetime by Cauchy surfaces $\Sigma_t$ of constant
Boyer-Lindquist time $t$. Our time flow vector is taken to be the Kerr
time-translation Killing vector field $t^a$, which is related to the
type D Killing field by $t^a = M^{1/3} \xi^a$. These satisfy
$t^a \nabla_a t = 1$. This gives the $3+1$ decomposition of the
metric,
\begin{equation}\label{eq:g 3+1}
g^{ab} = \sm^{ab} +\nv^a \nv^b ,
\end{equation}
where  $\sm_{ab}$ is a negative definite spatial metric intrinsic to $\Sigma_t$, and 
$
\nv^{a} = \frac{1}{N} \left(t^a - N^a \right)
$
is the unit surface normal which defines the 
lapse $N$ and shift $N^a$ of $t^a$. In our sign conventions $N=t_a \nv^a$ and
$N^a=\sm^{a}_{\ b}t^b$. We also introduce the Lagrangian density via
$L = \mathscr L \boldsymbol{e}$, where $e_{abcd}$ is a fixed
time-independent coordinate volume element. See appendix~E
of~\cite{Wald1984} for further details.

We work in a tetrad adapted to the Kerr principal null directions so
that the $\GHPLie_t$ derivative annihilates the background
quantities. We take this to be our ``time derivative'' and take the
Legendre transform with respect to
$\dot \Upsilon \equiv \GHPLie_t\Upsilon$ and
$\dot {\tilde \Upsilon} \equiv \GHPLie_t \tilde \Upsilon$. The
associated canonical momenta are
\begin{align}\label{eq:Tuek momenta}
  \varpi &= \frac{\partial \mathscr L}{\partial \dot {\tilde \Upsilon}} = \sqrt{-\sm} \,\nv^a\left(\Theta_a-4B_a\right)\Upsilon, \\
  \tilde \varpi &= \frac{\partial \mathscr L}{\partial \dot \Upsilon} = \sqrt{-\sm} \,\nv^a\left(\Theta_a+4B_a\right)\tilde \Upsilon.
\end{align}
Note the slightly nonstandard convention where we take the derivative
with respect to the conjugated field rather than the unconjugated
field to define the canonical momentum. This is so that a field and
its conjugate momentum have the same GHP weight, and it will be
convenient when we write the equations in first order form.

Finally, the Hamiltonian density is given by
\begin{align}\label{eq:Hamiltonian}
  \mathscr H \equiv{}& \pU \dot{\tilde \Upsilon} + \tilde \pU \dot\Upsilon-\mathscr L \nonumber\\
  ={}& \frac{N}{\sqrt{-\sm}}\pU \tilde \pU + \pU \left( 2 M^{1/3} (\Psi_2^{2/3} - 2 B^a\xi_a) + N^a(\Theta_a+4B_a)\right)\tilde \Upsilon
       +\tilde \pU \left(-2 M^{1/3} (\Psi_2^{2/3} - 2 B^a\xi_a) + N^a(\Theta_a-4B_a)\right)\Upsilon \nonumber\\
                     &- N \sqrt{-\sm}\left(
                       \sm^{ab}(\Theta_a+4B_a)\tilde \Upsilon(\Theta_b-4B_b)\Upsilon + 16\Psi_2\Upsilon\tilde \Upsilon\right),
\end{align}
with Hamiltonian $H = \int_{\Sigma_t} \mathscr{H}$.  This gives rise
to Hamilton's equations of motion,
\begin{align}
  \dot \Upsilon ={}& \frac{\delta H}{\delta \tilde \pU} 
  = \frac{N}{\sqrt{-\sm}} \pU - 2 M^{1/3} (\Psi_2^{2/3} - 2 \xi^a B_a)\Upsilon + N^a(\Theta_a - 4 B_a)\Upsilon, \label{eq:dotUpsilon} \\
  \dot \pU ={}& - \frac{\delta H}{\delta \tilde \Upsilon} \nonumber\\
  ={}& - \sqrt{-\sm}\left\{\sm^{ab}(\Theta_a - 4B_a)\left[ N(\Theta_b - 4B_b)\Upsilon\right] - 16 N \Psi_2 \Upsilon\right\}  - 2 M^{1/3} (\Psi_2^{2/3} - 2 \xi^a B_a)\pU + (\Theta_a - 4 B_a)(N^a\pU), \label{eq:dotvarpi}
\end{align}
as well as corresponding equations for conjugated fields.

It is convenient to also express these equations in matrix form,
\begin{equation}
  \GHPLie_t \boldsymbol{Y} = \mathcal{H} \boldsymbol{Y}, \qquad \boldsymbol{Y} \equiv \begin{pmatrix} \Upsilon \\ \varpi \end{pmatrix},
\end{equation}
where
%\begin{widetext}
\begin{align}\label{eq:Hmatrix}
  & \mathcal{H} = 
  \begin{pmatrix}
    s M^{1/3} (\Psi_2^{2/3} - 2 \xi^a B_a) + N^a(\Theta_a + 2s B_a)
    & \frac{N}{\sqrt{-\sm}} \\
    - \sqrt{-\sm}\left[\sm^{ab}(\Theta_a + 2s B_a) N(\Theta_b + 2s B_b) - 4s^2 N \Psi_2\right]
    & s M^{1/3} (\Psi_2^{2/3} - 2 \xi^a B_a) + (\Theta_a + 2s B_a)N^a
  \end{pmatrix}.
\end{align}
%\end{widetext}
Note that the derivative operators act on everything to the right;
here $s$ should again be viewed in the operator sense, as the weight
of the field on which $\mathcal H$ acts (in this case $s=-2$). The
equation for the conjugated fields has the same operator (but now with
$s=2$), $\GHPLie_t \tilde{\boldsymbol{Y}} = \mathcal{H} \tilde{\boldsymbol{Y}}$.

\section{Kerr geometry}
\label{app:D}
The Kerr metric is an asymptotically flat non-extremal, rotating black hole
spacetime for the values of the parameters $M>|a|$, assumed throughout the text. In Boyer-Lindquist coordinates, the Kerr metric takes the form
\begin{align}\label{eq:BL Kerr met}
g = &\left(1-\frac{2Mr}{\Sigma}\right) \dd t^2 + \frac{4Mar\sin^2\theta}{\Sigma} \dd t\dd\phi \nonumber \\&- \frac{\Sigma}{\Delta}\dd r^2 - \Sigma \dd\theta^2 - \frac{\Lambda}{\Sigma} \sin^2\theta \dd\phi^2,
\end{align} 
where 
\begin{align}\label{eq:BL stuff} 
  &\Delta=r^2+a^2-2Mr, \qquad \Sigma=r^2+a^2\cos^2\theta, \nonumber \\
  &\Lambda = (r^2 + a^2)^2 - \Delta a^2 \sin^2\theta.
\end{align}
We usually refer to the exterior of the Kerr manifold defined by $r>r_+$, 
with $r_+$ (event horizon) the greater root $r_\pm$ of $\Delta$. Sometimes we also refer to the tortoise coordinate
\begin{equation}\label{eq:rstar}
r_*=r+\frac{r_+^2+a^2}{r_+-r_-}\ln \left(\frac{r-r_+}{r_+}\right) - \frac{r_-^2+a^2}{r_+-r_-}\ln \left(\frac{r-r_-}{r_+}\right).
\end{equation}
The Kerr geometry has two commuting continuous symmetries generated by
the Killing fields
\begin{equation}
t^a = (\partial/\partial t)^a, \qquad \varphi^a = (\partial / \partial \phi)^a,
\end{equation}
The Kerr spacetime is Petrov type D, and therefore has two repeated
principal null directions. There is some freedom in choosing a NP
tetrad aligned with these null directions, but we will find it
convenient to choose the Kinnersley tetrad in explicit
calculations~\cite{Kinnersley:1969zza}. This tetrad is given in the above coordinates $(t,r,\theta,\phi)$ by
\begin{subequations}\label{eq:Kintet}
  \begin{align}
    l^a &= \frac{1}{\Delta} \left( r^2+a^2, \Delta, 0, a \right), \\
    n^a &= \frac{1}{2 \Sigma} \left( r^2+a^2, -\Delta,0,a \right), \\
    m^a &= \frac{1}{\sqrt{2} (r+ia\cos\theta)} \left( ia\sin\theta, 0,1,i\csc\theta \right). 
  \end{align}
\end{subequations}
The Kinnersley tetrad is regular on the past horizon of the black
hole. This choice of tetrad also sets the spin coefficient
$\epsilon=0$; the remaining non-zero spin coefficients appearing in this paper are
\begin{align}\label{eq:rhos and taus}
  \rho = - \frac{1}{r - ia\cos\theta},  
  \qquad   \rho' = -\frac{\rho\Delta}{2\Sigma},
  \qquad  \tau = -\frac{ia\sin\theta}{\sqrt 2 \Sigma},
  \qquad \tau' = -\frac{ia\rho^2\sin\theta}{\sqrt 2}.
\end{align}
The only nonzero Weyl scalar is
\begin{equation}\label{eq:Psi2 BL-K}
  \Psi_2 = - \frac{M}{(r-ia\cos\theta)^3} = M \rho^3.
\end{equation}
\end{widetext}
\bibliography{MyReferences.bib}
\end{document}